\documentclass[10pt,journal,compsoc]{IEEEtran}
\usepackage{caption} 
\usepackage{graphicx}
\usepackage{algorithm}
\usepackage{algorithmic}
\usepackage{amsmath}
\usepackage{bm}
\usepackage{amsthm}
\usepackage{multicol}
\usepackage{multirow}
\usepackage{color}
\usepackage{subfigure} 
\usepackage{threeparttable}
\usepackage{listings}
\usepackage{amsfonts}
\usepackage[numbers,sort&compress]{natbib}
\usepackage{longtable}
\usepackage{rotating}
\usepackage{stfloats}
\usepackage{ragged2e} 
\usepackage{color}
\usepackage{booktabs} 

\newtheorem{thm}{Theorem}[section]

\newtheorem{lemma}{Lemma}[section]
\makeatletter
\newcommand{\biggg}{\bBigg@{3}}
\makeatother
\begin{document}
\title{A Bilateral Game Approach for Task Outsourcing in Multi-access Edge Computing}
\author{Zheng Xiao*, \emph{Member}, \emph{IEEE},
        Dan He,
        Yu Chen,
        Anthony Theodore Chronopoulos, \emph{Senior member}, \emph{IEEE},
        Schahram Dustdar, \emph{Fellow}, \emph{IEEE},
        and Jiayi Du, \emph{Member}, \emph{IEEE}
\IEEEcompsocitemizethanks{
\IEEEcompsocthanksitem *Corresponding Author
\IEEEcompsocthanksitem Zheng Xiao, Dan He, Yu Chen and Jiayi Du are with College of Computer Science and Electronic Engineering, Hunan University, Hunan, China, 410082. E-mail: {\{zxiao, danhe\}@hnu.edu.cn, cy947205926@163.com, maxdujiayi@hnu.edu.cn,}.
\IEEEcompsocthanksitem Anthony Theodore Chronopoulos is with Department of Computer Science, University of Texas, San Antonio, TX, USA 78249 and Dept Computer Engineering Informatics 26500 Rio, University of Patras, Greece. Email: antony.tc@gmail.com.
\IEEEcompsocthanksitem Schahram Dustdar is with Distributed Systems Group, Vienna University of Technology, Vienna, Austria. Email: dustdar@dsg.tuwien.ac.at.
}
}

\markboth{}%
{Shell \MakeLowercase{\textit{et al.}}: Bare Demo of IEEEtran.cls for Computer Society Journals}
\IEEEtitleabstractindextext{%

\begin{abstract}
\justifying
    Multi-access edge computing (MEC) is a promising architecture to provide low-latency applications for future Internet of Things (IoT)-based network systems.
    Together with the increasing scholarly attention on task offloading, the problem of edge servers' resource allocation has been widely studied.
    Most of previous works focus on a single edge server (ES) serving multiple terminal entities (TEs), which restricts their access to sufficient resources.
    In this paper, we consider a MEC resource transaction market with multiple ESs and multiple TEs, which are interdependent and mutually influence each other.
    However, this many-to-many interaction requires resolving several problems, including task allocation, TEs' selection on ESs and conflicting interests of both parties.
    Game theory can be used as an effective tool to realize the interests of two or more conflicting individuals in the trading market.
    Therefore, we propose a bilateral game framework among multiple ESs and multiple TEs by modeling the task outsourcing problem as two noncooperative games: the supplier and customer side games. 
    In the first game, the supply function bidding mechanism is employed to model the ESs' profit maximization problem.
    The ESs submit their bids to the scheduler, where the computing service price is computed and sent to the TEs.
    While in the second game, TEs determine the optimal demand profiles according to ESs' bids to maximize their payoff.
    The existence and uniqueness of the Nash equilibrium in the aforementioned games are proved.
    A distributed task outsourcing algorithm (\textbf{\emph{DTOA}}) is designed to determine the equilibrium.
    Simulation results have demonstrated the superior performance of \textbf{\emph{DTOA}} in increasing the ESs' profit and TEs' payoff, as well as flattening the peak and off-peak load.
\end{abstract}
}

\maketitle
\IEEEdisplaynontitleabstractindextext
\IEEEpeerreviewmaketitle
\IEEEraisesectionheading{\section{Introduction}}
The Internet of Things (IoT) is a system of interrelated tens of billions of resource-hungry terminal entities (TEs), such as, sensors, wearable devices and unmanned aerial vehicles, which transfer data over a network with little or no human intervention.
With the development of TEs and wireless networks, the demand for low-latency computing services has been growing exponentially.
This paves the way for the development of Multi-access edge computing (MEC).

Multi-access edge computing (MEC) enables a powerful cloud at the edge of the network. MEC decentralizes networks and allows any enterprise or mobile operator to place a cloud at the edge, adjacent to the user.
In the MEC paradigm, plenty of machines are placed at the edge of the network so that computing services can be deployed on them for fast execution~\cite{Shuiguang2020Optimal}.
The MEC locates edge servers (ESs) with limited storage and computing resources at the edge of networks.
Since the computation capabilities and battery lives of TEs are limited, TEs offload computationally intensive tasks (e.g., program execution) to ESs (e.g., 4G/5G base stations).
The ESs execute these offloaded tasks and return results to TEs.
However, to take full advantage of these available computational resources of ESs, TEs' tasks need to be allocated appropriately~\cite{Cosmin2019Decentralized} .

Due to the resource constraints of ESs, some computationally intensive tasks will be offloaded to cloud servers (CSs), which are normally distantly located.
In that case, a higher transmission latency may be generated, which seriously degrades the quality of service (QoS).
Moreover, task offloading from ESs to CSs incurs extra latency and energy consumption due to communication between TEs and CSs.
Some important problems to be solved are how to satisfy the latency requirements of TEs and reduce the energy consumption in MEC.
In order to meet a strict QoS, Nafiseh \emph{et al.}~\cite{Nafiseh2019QoS} proposed a two-sided matching mechanism for edge services considering QoS requirements in terms of service response time.

In order to achieve low-latency and energy-saving computations, several studies on task offloading in MEC have been proposed. Xinchen Lyu \emph{et al.}~\cite{lyu2018selective} have attempted to encapsulate the latency requirements in offloading tasks and designed a selective offloading scheme. This scheme is achieved by enabling the devices to be self-denied or self-nominated for offloading. This can save energy consumption and minimize delays for task offloading. 
Authors in~\cite{zhao2015cooperative} developed a threshold-based strategy to improve the QoS, which combines the advantages of ESs' with lower latency and abundant computational resources of CSs. A priority queue is also applied to solve the delay problem, wherein delay-sensitive tasks are executed ahead of delay-tolerant tasks. 

However, most previous studies examined a single ES serving multiple TEs.
In fact, the proliferation of applications puts a heavy load on the ESs.
Since the ES has limited computing capacity and can't accommodate enough tasks, some task processes may have to wait longer to wake up.
Thus, for a single ES is hard to cope with concurrent tasks of multiple TEs, which will hinder the development and popularity of MEC.
In the future MEC market, there will be multiple different ESs offering optional computing service to TEs. Hence, TEs can choose different ESs according to their real-time and cost requirements.
In that case, multiple ESs provide computing services in parallel, which can accelerate the speed of task processing and alleviate offloading delays.


Task outsourcing has been employed as an effective paradigm by accommodating as many on-demand tasks as possible.
Its principle is to distribute TEs' tasks in different time slots according to ESs' load of each time slot.
Our work focuses on the problem of offloaded tasks outsourcing.
It differs from most existing works which focus on the necessity of offloading or on selecting which tasks to be offloaded to ESs or CSs.
The TEs' offloaded tasks are mapped to ESs according to their resource capacities.
Thus, task outsourcing can enhance scalability of MEC and satisfy TEs' dynamic service demands.

Fig.~\ref{fig:IoT} describes a MEC communication network in IoT, where ESs are deployed densely near TEs.
The MEC network is similar to a real competitive market, in which a wide range of TEs can be grouped into virtual clusters and compete for the ESs' limited wireless resources. 
They are interdependent and mutually influence each other.
Several base stations with ESs also compete with each other to win more TEs.
The ESs can communicate with the TEs via scheduler and inform them of their real-time service prices. 
Through the scheduler, TEs can participate in ESs selection, and make wise decisions regarding their daily computing resources consumption.

\begin{figure}[htbp]
	\centering
	\includegraphics[width=1.0\linewidth]{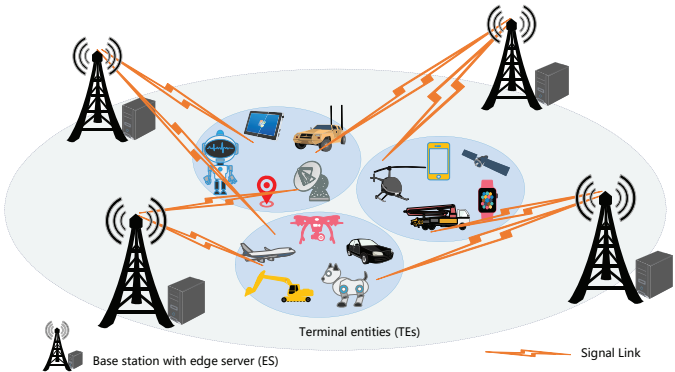}
	\caption{The MEC network scenario in IoT.} \label{fig:IoT}
\end{figure}

However, one of the main challenges of task outsourcing is to consider the interests of both parties.
On the one hand, ESs aim to get more profits, and strive for attracting more TEs to use their computing resources.
On the other hand, a rational TE will choose a task strategy that maximizes its own payoff.
Game theory can be used as an effective tool to model the interests of two or more conflicting individuals in trading markets and also load balancing in distributed systems~\cite{Grosu2005Noncooperative}.
The game solution was proved to be a Nash equilibrium solution for the noncooperative game.

In this paper, we propose a bilateral game framework among multiple ESs (the suppliers) and multiple TEs (the customers) to model the task outsourcing problem as two noncooperative games.
These two games are related to each other and are played simultaneously.
In the first game, the supply function bidding mechanism is employed to model the noncooperative game among ESs. 
In the proposed scheme, each ES, with limited or idle resources, submits a bid to reveal the available capacity ``supplied'' to the market.  
Then the scheduler collects these bids and computes a service price to clear the market so that the supply of the resource to be traded equals the demand.
In particular, all TEs are charged the same service price at one time slot.
The scheme can maximize the profits of ESs.
In the second game, in order to reduce costs, TEs determine the amount of assigned tasks for each time slot
based on the price of that time slot. 
If the price of one time slot is high, there would be fewer tasks assigned,
and if the price is low, there would be more tasks.
This framework can encourage TEs to assign fewer tasks during peak times or shift some tasks to off-peak times, which flattens the demand curve by peak clipping or valley filling. 



In summary, the contributions of this paper are:
 \begin{itemize}
	\item A bilateral-game framework is developed to model the interactions among TEs and ESs.
	\item A supply function bidding mechanism is proposed, where each ES submits a bid to reveal the available capacity ``supplied'' to the market.
	\item A \textbf{\emph{DTOA}} is designed to compute the Nash equilibrium.
	\item Simulations show that the proposed mechanism achieves the maximization of bilateral interests.
\end{itemize}

The remainder of this paper is organized as follows.
In Section~\ref{sec:related work}, the related work about task outsourcing in MEC is introduced.
Section~\ref{sec:System Model} models the task outsourcing problem in the ESs' and TEs' sides.
In Section~\ref{sec:Algorithm}, a \textbf{\emph{DTOA}} is designed to compute the Nash equilibrium in both sides.
Section~\ref{sec:Performance Evaluation} presents simulations showing the performance of the new approach using \textbf{\emph{DTOA}}.
Finally, conclusions are presented in Section~\ref{sec:Conclusion}.

\section{Related work}
\label{sec:related work}
In recent years, significant attention has been devoted to the resource allocation in MEC networks~\cite{jiang2016a}.
Shi Yan \emph{et al.}~\cite{yan2018game} studied the access selection for unmanned aerial vehicles (UAV) and bandwidth allocation of the base station (BS) in a UAV assisted IoT communication network. Wherein, the access competition among groups of UAVs is modeled as a dynamical evolutionary game.
The bandwidth allocation of BSs is formulated as a noncooperative game.
Authors in~\cite{you2016energy} examined resource allocation for a multi-TE MEC offloading system, which is formulated as a convex optimization problem for minimizing the weighted sum of mobile energy consumption.

Chunlin Li \emph{et al.}~\cite{li2019radio} analysed a radio and computing resource allocation problem between an access point and multiple devices in MEC system.
They designed a time average computation rate maximization algorithm to determine the optimal transmit power, and time allocation for the wireless devices.
Junhui Zhao \emph{et al.}~\cite{zhao2019computation} studied a cloud-MEC collaborative computation offloading problem that offloads tasks to automobiles in MEC vehicular networks.
They developed a tasks allocation optimization and collaborative computation offloading scheme to decide the optimal strategies.
An offloading algorithm for shortening the computation time and increasing the system utility was also designed.
Yunlong Gao \emph{et al.}~\cite{gao2019optimal} studied the optimal tradeoff between resource consumption and user experience in designing MEC systems.
Cosmin Avasalcai \emph{et al.}~\cite{avasalcai2019latency} introduced a decentralized resource management algorithm with the purpose of deploying IoT applications at the edge of the network such that end-to-end delay is minimized.

Since the resources of an ES are limited, the ES can not undertake the tasks coming from multiple TEs.
So multiple ESs are needed.
Nevertheless, the matching problem between multiple ESs and multiple TEs becomes a key issue.
Heli Zhang \emph{et al.}~\cite{zhang2017combinational} modeled the matching relationship between ESs and TEs as a commodity trading by applying a multi-round sealed sequential combinational auction mechanism.
In~\cite{gu2019task}, the authors studied task offloading in vehicular MEC environments and modelled the interactions between edges and tasks as a matching game.
They further developed two standalone heuristic algorithms to minimize the average delay while taking the energy consumption and vehicle mobility constraints into consideration.
A three-tier IoT fog network was proposed in~\cite{zhang2017computing}, in which all fog nodes, data service operators and data service subscribers are jointly optimized to achieve the optimal resource allocation in a distributed fashion.

Furthermore, authors in~\cite{liu2017price,tang2019jointly} adopted a price-based mechanism to design efficient resource allocation in a MEC network. For example,~\cite{liu2017price} proposed a price-based distributed method to manage the offloaded tasks from users. Wherein, edge cloud sets prices to maximize its revenue and each user makes an optimal decision to minimize her/his own cost.
The work~\cite{tang2019jointly} proposed a price-based resource allocation mechanism among the MEC server and multiple base stations (BSs). The MEC server tries to provide prices to BSs so as to maximize its own revenue while the BSs determine the computing space to improve the quality of experience.

To summarize the related work above, we observe that the existing resource allocation and matching problem in MEC generally involves edge nodes and clients using resource from an edge node.
However, most of these studies focus either on the system performance or ESs' benefits, while ignoring the TEs' pursuit of maximizing payoff.
Against this backdrop, our paper tries to balance the objectives of both ESs and TEs. 
In this paper, we also adopt a priced-based supply bidding mechanism to solve the resource allocation problem. 

\section{System Model}
\label{sec:System Model}

\subsection{Interaction between TEs and ESs} 
\begin{figure}[htbp]
	\centering
	\includegraphics[width=0.9\linewidth]{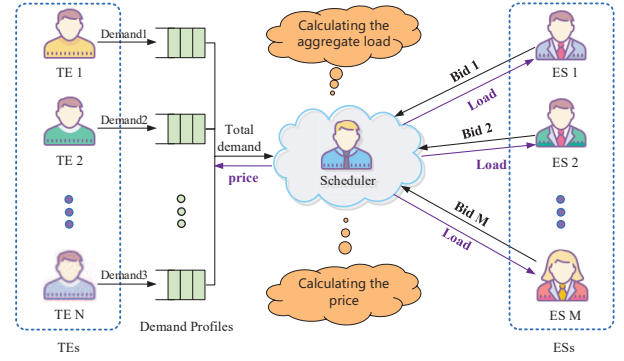}
	\caption{Diagram of a resources transaction market.} \label{fig:Diagram}
\end{figure}
As shown in Fig. \ref{fig:Diagram}, we consider a scheduler-based resources transaction market, which consists of $M$ ESs and $N$ TEs in a MEC network.
ESs act as suppliers who sell computing resources to TEs and TEs act as customers who purchase resources from ESs.
The bidirectional interaction between ESs and TEs is performed through a scheduler, which serves as a third-party agency outsourcing TEs' tasks to ESs.

On the one hand, TEs submit their demand profiles to the scheduler via a communication network. 
On the other hand, ESs compete with each other for acquiring more TEs, and submit bids based on strategies of their opponents and their own resource capacities.
As a response, the scheduler calculates the service price and the aggregated load based on ESs' bids and TEs' total demand.
They are mutually dependent upon each other as decisions on either side can have a bearing on those of the other side.
After receiving the real-time price signal, the TEs will update their demand profiles.
Since the aggregate load depends on the TEs' demand profiles, the behavior of TEs will affect the ESs' bidding strategies. 
The aforementioned process is repeated until both customers and suppliers are satisfied.

We divide one day into a set of $T$ ($T=24$) time slots, denoted as $\mathcal{T} = \{ 1,\cdots,T \}$.
The set of ESs and TEs are represented as $\mathcal{M}= \{ 1,\cdots,M \}$ and $\mathcal{N}=\{1, \ldots, N\}$.
How many resources should ESs provide to the market and how the ESs' bids affect the TEs' demand profiles are questions worth investigation.
We next present the model of both sides in the MEC resources transaction market.

\subsection{Cost and profit of ES}

For ES $j \in \mathcal{M}$, let $C_{j,t}(.)$ denote the cost function of ES $j$ at time slot $t \enspace (t \in \mathcal{T})$.
Let $R_{j,t}(.)$ denote the revenue function of ES $j$ at the $t$th time slot by providing the computational load.
The profit of ES equals the revenue by providing computing service minus its cost of system overhead. 
Therefore, the profit $P_{j,t}$ of ES $j$ at time slot $t$ can expressed as follows:
\begin{equation}
\label{equ:P_j,t}
	P_{j,t}=R_{j,t}(.)-C_{j,t}(.).
\end{equation} 

We consider that each ES is selfish and tries to maximize its own profit. Thus, the interaction among the profit maximizer ESs can be modeled as a noncooperative game.
The ESs are the players while the bid profiles are the strategies.
Let ${\lambda}_{j,t}$ denote the bid of ES $j$ at time slot $t$.
The target of each ES $j$ is to find the optimal bid ${\lambda}_{j,t}$ to maximize its profit, which can be defined as:
\begin{equation}
\label{equ:maxP} 
	\underset{\lambda_{j, t}}{\operatorname{maximize}} \quad P_{j,t} \quad { j \in \mathcal{M}, t \in \mathcal{T}}.
\end{equation}
By substituting Equ. (\ref{equ:P_j,t}) into Equ. (\ref{equ:maxP}), we can get
\begin{equation}
\label{equ:maxR-C} 
	\underset{\lambda_{j, t}}{\operatorname{maximize}} \quad R_{j,t}(.)-C_{j,t}(.) \quad { j \in \mathcal{M}, t \in \mathcal{T}}.
\end{equation}

We denote by $f_{j, t}$ the task load that ES $j$ willing to generate in the time slot $t$.
We assume that the service price of different ESs in one time slot is the same and denoted as $p_{e}(t)$ at time slot $t$.
The revenue of each ES is equal to the product of its load and the service price.
Hence, the revenue of ES $j$ at time slot $t$ can be represented as
\begin{equation}
\label{equ:R_jt}
R_{j,t}=f_{j, t} \cdot p_{e}(t).
\end{equation}

Similar to~\cite{jalali2015demand}, the ES $j$'s cost function is defined as a quadratic function
\begin{equation*}
\label{equ:C_jt}
C_{j,t}\left(f_{j, t}\right)=a_{j,2} f_{j, t}^{2}+a_{j,1} f_{j, t}+a_{j,0},
\end{equation*}
where $a_{j,2}$, $a_{j,1}$ and $a_{j,0}$ are positive coefficients and model the fact that different ESs incur different costs for serving the tasks.
We note that the cost function is increasing and convex.
Substituting Equ. (\ref{equ:R_jt}) into Equ. (\ref{equ:maxR-C}), the optimization problem can be further rewritten as
\begin{equation}
\label{equ:maxR-C1} 
	\begin{aligned}
		& \underset{\lambda_{j, t}}{\operatorname{maximize}} \quad f_{j, t} \cdot p_{e}(t)-C_{j,t}(f_{j,t}) \\ 
		& {\text { subject to }} \quad {f_{j, t} \geq 0, \quad j \in \mathcal{M}, t \in \mathcal{T}}.
	\end{aligned}
\end{equation}

\subsection{Payoff and payout of TE}
The demand of each TE consists of two parts: a base demand and a shiftable demand.
On the one hand, a base demand is primarily concerned with real-time tasks, which have high priority.
On the other hand, a shiftable demand has low priority real-time requirements and it can be assigned at any time slot.
The shiftable demand profile of TE $i \enspace (i \in \mathcal{N})$ is defined as $\bm{\chi}_{i}=\left(\chi_{i,1}, \ldots, \chi_{i,T}\right)$ and the base demand of TE $i$ at time slot $t$ is denoted as $r_{i,t}$, which is known and fixed.

The utility of TE $i$ represents the profit that TE $i$ receives when it completes tasks and is denoted as $U_{i}(.)$.
Exactly, the utility function of TE $i$ is the utility for the tasks rather than the service time or applications.
Similar to~\cite{samadi2010optimal}, we employ the quadratic utility function because it is non-decreasing and its marginal benefit is non-decreasing,
\begin{equation}
\label{equ:U}
	U_{i}(x)=
		\left\{
			\begin{aligned}
				& {w_{i, t} x-\frac{\alpha_{i, t}}{2} x^{2},} \quad {0 \leq x \leq \frac{w_{i, t}}{\alpha_{i, t}}} \\
				& {\frac{w_{i, t}^{2}}{2 \alpha_{i, t}},} \quad \quad \quad \quad \quad {x > \frac{w_{i, t}}{\alpha_{i, t}}}
			\end{aligned}\right.,
\end{equation}
where $x=(\chi_{i,t}+r_{i,t})$, $w_{i, t}$ and $\alpha_{i, t}, i \in \mathcal{N}$ are coefficients that reflects the dynamic changes of TE $i$'s demand.

The payout function quantifies the payout that TE $i$ needs to pay the ESs task completion.
Without loss of generality, we define the payout of TE $i$' as the product of demand and the service, i.e.
\begin{equation}
\label{equ:payout}
	Payout_{i,t}=(\chi_{i,t}+r_{i,t}) \cdot p_{e} (\bm{\lambda}_{t}, L_{t}).
\end{equation}
The payoff function quantifies the final benefits of TE $i$ and represents the satisfaction of using the service.
Thus, we denote the payoff of TE $i$ as its utility minus payout i.e.
\begin{equation}
\label{equ:payoff}	
	Payoff_{i} = Utility_{i} - Payout_{i}.
\end{equation}
Let $u_{i}$ denote the payoff of TE $i$.
By substituting Equ. (\ref{equ:U}) and Equ. (\ref{equ:payout}) into Equ. (\ref{equ:payoff}), we can obtain
\begin{equation}
\label{equ:u}
	\begin{aligned}
		u_{i}\left(\bm{\chi}_{i}, \bm{\chi}_{-i}\right)=\sum_{t \in \mathcal{T}}\Big(&U_{i}\left(\chi_{i,t}+r_{i,t}\right)\\ &-\left(\chi_{i,t}+r_{i,t}\right) p_{e} \left(\bm{\lambda}_{t}, L_{t}\right)\Big),
	\end{aligned}	
\end{equation}
where $\bm{\chi}_{-i}$ denotes the vector of the demand profile of other TEs and $\bm{\chi}_{-i}=\left(\bm{\chi}_{1}, \dots, \bm{\chi}_{i-1}, \bm{\chi}_{i+1}, \ldots, \bm{\chi}_{N}\right)$.  In Equ. (\ref{equ:u}), the utility is a function related to $(\chi_{i,t}+r_{i,t})$.

Each TE tries to maximize its payoff by determining its shiftable demand profile.
Thus, the interaction between TEs can be modeled as a noncooperative game.
The TEs are participants while the shiftable demand profiles are the strategies of the noncooperative game.

Let $\bm{\chi}_{i}^{*}$ denote the optimal demand profile of TE $i$ in the Nash equilibrium and $Q_{i}^{total}$ denote the total daily shiftable demand of TE $i$ which is fixed and known.
Let $L_{t}$ denote the aggregate load demand of the ESs at time slot $t$ and $L_{t}=\sum_{j \in \mathcal{N}}\left(\chi_{j,t}+r_{j,t}\right)$. Considering TE $i$, the optimization problem can be formulated as follows when other TEs' profiles are fixed:
\begin{equation}
\label{equ:max_u}
	\begin{aligned}
		& \underset{\bm{\chi}_{i}}{\operatorname{maximize}} \quad
        u_{i}\left(\bm{\chi}_{i}, \bm{\chi}_{-i}\right) \\
        & {\text { subject to }} \quad \sum_{t \in \mathcal{T}} \chi_{i, t}=Q_{i}^{\text {total}}, \\
        &\quad \quad \quad \quad \quad \quad \chi_{i, t} \geq 0, \forall i \in \mathcal{N}.
	\end{aligned}
\end{equation}


\subsection{Market mechanism with supply function bidding}
In this section, we employ a supply function bidding mechanism to model the relationship between market demand for services and its price.
We use a class of supply functions with parameters.
The bids submitted by ESs reveal their available resource capacities ``supplied'' to the market. 

\begin{table}[htbp]
	\caption{Definitions of Mathematical Notations}
	\label{tab:1}
	\begin{tabular}{@{}|c|l|@{}}
		\toprule
		\textbf{Notation}&\textbf{Definition}\\ 
		\midrule
		$L_t$&TEs' total load demand at time slot $t$\\
		\midrule 
		$f_{j,t}$&The supply function of the ES $j$ at time slot $t$\\
		\midrule
		$p_{e}(t)$&The computing service price at time slot $t$\\
		\midrule
		$p_{1},\cdots,p_{K}$&\begin{tabular}[c]{@{}l@{}}$K$ break points of the price-wise linear function of\\ all ESs\end{tabular} \\
		\midrule
		$\lambda_{j,t}^{k}$&\begin{tabular}[c]{@{}l@{}}The slope of the function between the break points\\ $p_{k-1}$ and $p_{k}$\end{tabular}\\
		\midrule
		$\lambda_{j,t}^{1}$&\begin{tabular}[c]{@{}l@{}}The slope of the function between the origin and\\ break point $p_{1}$\end{tabular}\\
		\bottomrule
	\end{tabular}
\end{table}
\begin{figure}[htbp]
	\centering
	\includegraphics[width=0.9\linewidth]{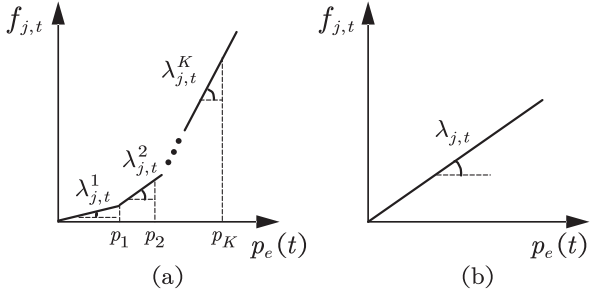}
	\caption{(a) Piece-wise linear. \enspace \enspace (b) Affine supply functions.} \label{fig:Affine}
\end{figure}
The notations used in the supplier side model are presented in Table \ref{tab:1}. We assume that the supply function $f_{j,t}$ is chosen from the family of increasing and convex price-wise linear functions of $p_{e}(t)$~\cite{kamyab2015demand}. 
Fig. \ref{fig:Affine}(a) shows an increasing and convex piece-wise linear supply function.
The abscissa $p_{e}(t)$ indicates the price and the ordinate $f_{j,t}$ denotes the load supplied by the TE $j$ at time slot $t$.
There exists $K$ break points on the abscissa of the Fig. \ref{fig:Affine}(a).
$\lambda_{j,t}^{k} \geq 0$ represents the slope of the function between the break points $p_{k-1}$ and $p_{k}$.
Fig. \ref{fig:Affine}(b) shows the affine supply function.

At time slot $t \enspace (t \in \mathcal{T})$, we use the vector $\bm{\lambda}_{j,t} = (\lambda_{j,t}^{1},\cdots,\lambda_{j,t}^{K})$ to denote the bid profile of ES $j \enspace (j \in \mathcal{M})$. Thus, we obtain
\begin{equation}
\label{equ:f1}
	f_{j, t}\left(p_{e}(t), \bm{\lambda}_{j, t}\right) \!=\!
	\left\{
		\begin{aligned}
			& \lambda_{j, t}^{1} p_{e}(t), \enspace 0 \leq p_{e}(t) \leq p_{1} \\\\
			& \lambda_{j, t}^{k} p_{e}(t) \!+\!\lambda_{j, t}^{k-1} p_{k-1}, \enspace p_{k-1} < p_{e}(t) \leq p_{k}
		\end{aligned}
	\right..
\end{equation}

It is assumed that each ES submits $\bm{\lambda}_{j,t}$ as a bid profile to the scheduler at time slot $t$.
For each ES $j$, the bid profile describes the number of tasks that it is willing to admit.
We can use $\bm{\lambda}_{t}$ to represent the bid profiles of all ESs at time slot $t$ and $\bm{\lambda}_{t}= \{ \bm{\lambda}_{1,t},\cdots,\bm{\lambda}_{M,t} \}$.
In response to ESs, the scheduler sets the price $p_{e}(t)$ to clear market.
In economics, market clearing means the supply of what is traded equals the demand, so that there is no leftover supply or demand.
In this case, the demand of all TEs is the same as the load supplied by all ESs.
Although the fluctuation in TEs' demand will drive changes in ESs' bid profiles, the demand and supply remains balanced.
The equivalence further builds up the connection between the supplier game and the customer game.
Hence, it can expressed as
\begin{equation}
\label{equ:f2}
		\sum_{j \in \mathcal{M}} f_{j, t}\left(p_{e}(t), \bm{\lambda}_{j, t}\right)=L_{t}, \quad t \in \mathcal{T}.
\end{equation}
According to Equ. (\ref{equ:f1}) and Equ. (\ref{equ:f2}), we have
\begin{equation}
\label{equ:Lt}
	L_{t} \!=\!
	\left\{
		\begin{aligned}
			& \sum_{j \in \mathcal{M}} \left( \lambda_{j, t}^{1} p_{e}(t)\right), \enspace 0 \leq p_{e}(t) \leq p_{1} \\\\
			& \sum_{j \in \mathcal{M}} \left( \lambda_{j, t}^{k} p_{e}(t) \!+\!\lambda_{j, t}^{k-1} p_{k-1}\right), \enspace p_{k-1} < p_{e}(t) \leq p_{k}
		\end{aligned}
	\right..
\end{equation}
According to Equ.~\ref{equ:Lt}, we can further calculate the service price function as follows:
\begin{equation}
\label{equ:pet}
p_{e}(t)\!=\!
	\left\{
		\begin{aligned}
			& \frac{L_{t}}{\sum_{j \in \mathcal{M}} \lambda_{j,t}^{1}}, \enspace 0 \leq p_{e}(t) \leq p_{1}\\\\
			& \frac{L_{t}-\sum_{j \in \mathcal{M}} \left( \lambda_{j, t}^{k-1} p_{k-1} \right)}
			{\sum_{j \in \mathcal{M}} \lambda_{j, t}^{k}}, \quad p_{k-1} < p_{e}(t) \leq p_{k}.			
		\end{aligned}
	\right..
\end{equation}
At time slot $t$, the service price of different ESs is the same.

In~\cite{baldick2004theory}, the affine supply function $f_{j, t}\left(p_{e}(t), \bm{\lambda}_{j, t}\right) =\lambda_{j, t}^{1} p_{e}(t)$ is used as a special case of the aforementioned piece-wise linear functions.
Almost all the results of affine supply functions can be generalized to the piece-price affine supply function~\cite{chen2006coevolutionary}.
As for Equ. (\ref{equ:pet}), it can be concluded that the affine function is equivalent to the piece-wise linear supply function between two break points.
Each piecewise function of Fig. \ref{fig:Affine}(a) can be regarded as a linear function in Fig. \ref{fig:Affine}(b).
As a matter of fact, the term $\lambda_{j, t}^{k-1} p_{k-1}$ is fixed when we are between break points $p_{k-1}$ and $p_{k}$.
Therefore, without loss of generality, we generalize the results from the affine functions to piece-wise linear functions. 
Therefore, the computing service price can be given as follows for an affine supply function:
\begin{equation}
\label{equ:p2}
	p_{e}(t)=\frac{L_{t}}{\sum_{j \in \mathcal{M}} \lambda_{j,t}^{1}}, \quad t \in \mathcal{T}.
\end{equation}

For simplicity, we use the notation $\lambda_{j,t}$ instead of $\lambda_{j,t}^{1}$ to represent the affine supply function of ES $j$. Meanwhile, we use $\bm{\lambda}_{t}=\left(\lambda_{1, t}, \ldots, \lambda_{M, t}\right)$ to denote the bids profile for all ESs at time slot $t$. As Equ. (\ref{equ:p2}) shows, the computing service price is related to $\lambda_{j, t} \enspace (j \in \mathcal{M})$ and $L_{t}$. Hence, the price function can be denoted as $p_{e} \left(\bm{\lambda}_{t}, L_{t}\right)$. As suggested by Equ. (\ref{equ:f1}), supply function $f_{j,t}$ for ES $j$ can be expressed as
\begin{equation}
\label{equ:f3} 
	f_{j, t}\left(p_{e} \left(\bm{\lambda}_{t}, L_{t}\right), \lambda_{j, t}\right)=\frac{\lambda_{j, t} L_{t}}{\sum_{r \in \mathcal{M}} \lambda_{r, t}}, \quad t \in \mathcal{T}.
\end{equation}

Similar to the computing service, the supply function can be represented by $f_{j, t}\left(\bm{\lambda}_{t}, L_{t}\right)$. Let $\bm{\lambda}_{-j,t}$ denote the submitted bids of other ESs except for ES $j$. So it can be defined as $\bm{\lambda}_{-j,t}=\left(\lambda_{1, t}, \ldots, \lambda_{j-1, t}, \lambda_{j+1, t}, \ldots, \lambda_{M, t}\right)$.
Hence, According to Equ. (\ref{equ:maxR-C1}) and Equ. (\ref{equ:f3}), the profit function of ES $j$ is rewritten as 
\begin{equation}
\label{equ:mu}
	P_{j,t}(\lambda_{j, t}, \bm{\lambda}_{-j, t}) = \frac{\lambda_{j, t} L_{t}^{2}}{\left( \sum_{r \in \mathcal{M}} \lambda_{r, t} \right)^{2}}-C_{j}\left(\frac{\lambda_{j, t} L_{t}}{\sum_{r \in \mathcal{M}} \lambda_{r, t}}\right).
\end{equation}
When other ESs' bids are fixed, the ES $j$ tries to find the optimal bid $\lambda_{j,t}^{*}$ by solving the following optimization problem:
\begin{equation}
\label{equ:max1} 
    \begin{aligned}
        & \underset{\lambda_{j, t}}{\operatorname{maximize}} \quad {\frac{\lambda_{j, t} L_{t}^{2}}{\left( \sum_{r \in \mathcal{M}} \lambda_{r, t} \right)^{2}}-C_{j}\left(\frac{\lambda_{j, t} L_{t}}{\sum_{r \in \mathcal{M}} \lambda_{r, t}}\right)} \\\\ 
        & {\text { subject to }} \quad {\lambda_{j, t} \geq 0, \quad j \in \mathcal{M}, t \in \mathcal{T}}.
    \end{aligned}
\end{equation}

\subsection{Nash equilibrium analysis}

The following section will explain that the ES's game (Equ. (\ref{equ:max1})) has a unique Nash equilibrium, as shown by the lemma below.
\begin{lemma}
\label{lem:1}
    Assume that the bids profile in Nash equilibrium at time slot $t$ is denoted as $\bm{\lambda}_{t}^{*}$. When the Nash equilibrium is reached, it will satisfy $\lambda_{j, t}^{*}<\sum_{r \in \mathcal{M}, r \neq j} \lambda_{r, t}^{*}$ for all ESs.
\end{lemma}
\begin{proof}
The function $\Pi_{j, t}\left(\lambda_{j, t}, \bm{\lambda}_{-j, t}\right)$ is expressed as follows:
\begin{equation}
\label{equ:Pi} 
	\Pi_{j, t}\left(\lambda_{j, t}, \bm{\lambda}_{-j, t}\right)=\frac{\lambda_{j, t} L_{t}^{2}}{\left(\sum_{r \in \mathcal{M}} \lambda_{r, t}\right)^{2}}.
\end{equation}

As the formula above suggests, $\Pi_{j, t}\left(\lambda_{j, t}, \bm{\lambda}_{-j, t}\right)$ is the first term in $P_{j, t}\left(\lambda_{j, t}, \bm{\lambda}_{-j, t}\right)$.
From Equ. (\ref{equ:Pi}), we can calculate the first derivative function as follows
\begin{equation}
\begin{aligned}
	&\frac{d \Pi_{j, t}\left(\lambda_{j, t}, \bm{\lambda}_{-j, t}\right)}{d \lambda_{j, t}}\\
	&=\frac{L_{t}^{2} \cdot\left(\sum_{r \in \mathcal{M}} \lambda_{r, t}\right)^{2}-2 \lambda_{j, t} L_{t}^{2}\left(\sum_{r \in \mathcal{M}} \lambda_{r, t}\right)}{\left(\sum_{r \in \mathcal{M}} \lambda_{r, t}\right)^{4}}
\end{aligned}
\end{equation}
Let
\begin{equation*}
\frac{d \Pi_{j, t}\left(\lambda_{j, t}, \bm{\lambda}_{-j, t}\right)}{d \lambda_{j, t}}>0,
\end{equation*}
we can get
\begin{equation}
\label{equ:dPi}
	\left(\sum_{r \in \mathcal{M}} \lambda_{r, t}\right)^{2}-2 \lambda_{j, t} \sum_{r \in \mathcal{M}} \lambda_{r, t}>0.
\end{equation}
The Equ. (\ref{equ:dPi}) is equivalent to
\begin{equation}
\label{equ:dPi1}
	\left(\lambda_{j, t}+\sum_{r \in \mathcal{M}, r \neq j} \lambda_{r, t}\right)^{2}-2 \lambda_{j, t} \left(\lambda_{j, t}+\sum_{r \in \mathcal{M}, r \neq j} \lambda_{r, t}\right)>0.
\end{equation}
From Equ. (\ref{equ:dPi1}), we can derive that
\begin{equation*}
	0 \leq \lambda_{j, t}<\sum_{r \in \mathcal{M}, r \neq j} \lambda_{r, t}.
\end{equation*}

In summary, we can conclude that $P_{j, t}\left(\lambda_{j, t}, \bm{\lambda}_{-j, t}\right)$ is an increasing function when $0 \leq \lambda_{j, t}<\sum_{r \in \mathcal{M}, r \neq j} \lambda_{r, t}$. And it becomes a decreasing function when $ \lambda_{j, t} \geq \sum_{r \in \mathcal{M}, r \neq j} \lambda_{r, t}$. Thus, in order to maximize profit, we should meet the constraint $0 \leq \lambda_{j, t}<\sum_{r \in \mathcal{M}, r \neq j} \lambda_{r, t}$. In the Nash equilibrium, the bid of ES $j$ at time slot $t$ is denoted as $\lambda_{j,t}^{*}$. Therefore, we can conclude that $\lambda_{j, t}^{*}<\sum_{r \in \mathcal{M}, r \neq j} \lambda_{r, t}^{*} \enspace (j \in \mathcal{M})$.
\end{proof}

Similar to~\cite{johari2006parameterized}, the proof for the following theorem given as follows.
\begin{thm}
\label{the:1}
	The ES's noncooperative game has a unique Nash equilibrium. Furthermore, the Nash equilibrium is the solution of the following convex optimization problem:
	\begin{equation}
	\label{equ:max2} 
		\begin{array}{cl}{
			\underset{0 \leq f_{j, t}<\frac{L_{t}}{2}}{\operatorname{maximize}}} & {\sum \limits_{j \in \mathcal{M}}-\Psi_{j}\left(f_{j, t}\right)} \\ \\
			{\text {subject to}} & {\sum \limits_{j \in \mathcal{M}} f_{j, t}=L_{t}},
		\end{array}
	\end{equation}
	where
	\begin{equation}
	\label{equ:Psi} 
		\Psi_{j}\left(s_{j, t}\right)=\left(\frac{L_{t}-f_{j, t}}{L_{t}-2 f_{j, t}}\right) C_{j}\left(f_{j, t}\right)-\int_{0}^{f_{j, t}} \frac{L_{t} C_{j}\left(\Pi_{j}\right)}{\left(L_{t}-2 \Pi_{j}\right)^{2}} d \Pi_{j}.
	\end{equation}
\end{thm}
\begin{proof}
	According to lemma \ref{lem:1}, we can infer that the load supplied by each ES at time slot $t$ does not exceed $L_{t}/2$ at the Nash equilibrium. The Lagrange function of the optimization problem in Equ. (\ref{equ:max2}) is denoted as $F$. Thus, we have
	\begin{equation}
	\label{equ:F} 
		F=\sum_{j \in \mathcal{M}}-\Psi_{j}\left(f_{j, t}\right)+\phi \left(\sum_{j \in \mathcal{M}} f_{j, t}-L_{t}\right),
	\end{equation}
	where $\phi$ denotes the Lagrange multiplier.
	We can obtain the following expression through the first-order optimality function.
	\begin{equation}
	\label{equ:F_d1} 
		\left(\frac{\partial F}{\partial f_{j, t}^{*}}\right)\left(f_{j, t}-f_{j, t}^{*}\right) \leq 0, \quad \forall j \in \mathcal{M},
	\end{equation}
	where $f_{j, t}^{*}$ is defined as the supply function in equilibrium, while $\phi^{*}$ is the Lagrange multiplier in equilibrium.

	From Equ. (\ref{equ:Psi}), $({\partial F}/{\partial f_{j, t}^{*}})$ can be expressed as follows:
	\begin{equation}
	\label{equ:F_d2}
		\frac{\partial F}{\partial f_{j, t}^{*}}=\phi^{*}-\left(\frac{L_{t}-f_{j, t}^{*}}{L_{t}-2 f_{j, t}^{*}}\right) C_{j}^{\prime}\left(f_{j, t}^{*}\right).
	\end{equation}
	
	We assume the first-order optimality condition for the optimization problem in Equ. (\ref{equ:max1}). Thus, we obtain
	\begin{equation}
	\label{equ:mu_d}
		\left(\frac{\partial P_{j, t}}{\partial \lambda_{j, t}}\right)\left(\lambda_{j, t}-\lambda_{j, t}^{*}\right) \leq 0, \quad \forall j \in \mathcal{M}.
	\end{equation}
	From Equ. (\ref{equ:mu}), ${\partial P_{j, t}}/{\partial \lambda_{j, t}}$ is calculated as follows:
	\begin{equation}
	\label{equ:pi2}
		\frac{\partial P_{j, t}}{\partial \lambda_{j, t}}=p_{e} \left(\bm{\lambda}_{t}, L_{t}\right)-\frac{L_{t}-f_{j, t}^{*}}{L_{t}-2 f_{j, t}^{*}} C_{j}^{\prime}\left(f_{j, t}^{*}\right).
	\end{equation}
	By substituting Equ. (\ref{equ:pi2}) into Equ. (\ref{equ:mu_d}), we can write the optimality condition for Nash equilibrium as follows
	\begin{equation}
	\label{equ:p3}
		\left(p_{e} \left(\bm{\lambda}_{t}, L_{t}\right)-\frac{L_{t}-f_{j, t}^{*}}{L_{t}-2 f_{j, t}^{*}} C_{j}^{\prime}\left(f_{j, t}^{*}\right)\right)\left(\lambda_{j, t}-\lambda_{j, t}^{*}\right) \leq 0.
	\end{equation}

	From Equ. (\ref{equ:F_d1}) and Equ. (\ref{equ:p3}), we can see that the Lagrange multiplier is actually the price $p_{e}(\bm{\lambda}_{t}, L_{t})$ of the computing service.
	In addition, the optimality condition Equ. (\ref{equ:F_d1}) is equivalent to Equ.\ref{equ:p3}.
	Therefore, the existence and uniqueness of the Nash equilibrium is equivalent to proving the existence and uniqueness of the optimal point of problem Equ. (\ref{equ:max2}).
\end{proof}

In Theorem \ref{lem:1}, it is proved that the ES's game has a unique Nash equilibrium solution, whose strategies are determined by the aggregate load $L_t$. Besides, a ES can scale-up and scale-down its resource capacity according to different market demands. Thus, ESs will bid differently for different levels of load.

We next analyze the existence of Nash equilibrium for the customer side game, which is proved by the theorem below.

\begin{thm}
	The customers' optimization problem is a convex programming problem. In fact, the customer side game Equ. (\ref{equ:max_u}) is an n-person game. It has a unique pure strategy Nash equilibrium.
\end{thm}
\begin{proof}
From the above discussion it follows that the objective function in Equ. (\ref{equ:max_u}) is equal to
\begin{equation}
\label{equ:max_u1}
	\begin{aligned}
		\sum_{t \in \mathcal{T}} & U_{i}\left(\chi_{i,t}+r_{i,t}\right)\\
		&-\sum_{t \in \mathcal{T}} \Bigg(\frac{\left(\chi_{i,t}+r_{i,t}\right)^{2} +\sum_{j \in \mathcal{N},j \neq i}\left(\chi_{j,t}+r_{j,t}\right)}{\sum_{r \in \mathcal{M}} \lambda_{r, t}}\Bigg).
	\end{aligned}
\end{equation}
Let $k=1/(\sum_{r \in \mathcal{M}} \lambda_{r, t}), k> 0 $. For simplicity, we denote the right part of Equ. (\ref{equ:max_u1}) as follows:
\begin{equation*}
\label{equ:hk}
	h_{i}(\bm{\chi}_{i},\bm{\chi}_{-i})=\sum_{t \in \mathcal{T}} k\bigg(\left(\chi_{i,t}+r_{i,t}\right)^{2} +\sum_{j \in \mathcal{N},j \neq i}\left(\chi_{j,t}+r_{j,t}\bigg)\right).
\end{equation*}
We have
\begin{equation}
\begin{aligned} 
	\nabla_{\bm{\chi}_{i}} h_{i}&(\bm{\chi}_{i},\bm{\chi}_{-i}) 
=\left[\frac{\partial h_{i}(\bm{\chi}_{i},\bm{\chi}_{-i})}{\partial \chi_{i,t}}\right]_{t=1}^{T} \\ 
&=\left(\frac{\partial h_{i}(\bm{\chi}_{i},\bm{\chi}_{-i})}{\partial \chi_{i,1}}, \cdots, \frac{\partial h_{i}(\bm{\chi}_{i},\bm{\chi}_{-i})}{\partial \chi_{i,T}}\right)\\
&=2k\left[\left(\chi_{i,t}+r_{i,t}\right)+\sum_{j \in \mathcal{N},j \neq i} \left(\chi_{j,t}+r_{j,t}\right)\right]_{t=1}^{T}\\
\end{aligned}
\end{equation}
and the Hessian matrix is as follows:
\begin{equation}
\nabla^{2}_{\bm{\chi}_{i}} h_{i}(\bm{\chi}_{i},\bm{\chi}_{-i})=\left(\begin{array}{cccc}{2k} & {2k} & {\cdots} & {2k} \\ {2k} & {2k} & {\cdots} & {2k} \\ {\vdots} & {\vdots} & {\ddots} & {\vdots} \\ {2k} & {2k} & {\cdots} & {2k}\end{array}\right)_{N \times T}.
\end{equation}
This further leads to
\begin{equation*}
\begin{array}{c}
	X^{\mathrm{T}} \nabla^{2}_{\bm{\chi}_{i}} h_{i}(\bm{\chi}_{i},\bm{\chi}_{-i}) X=2k \left(X_{1}+X_{2}+\cdots+X_{N \times T}\right)^{2} \geq 0, \\\\
	\quad \forall X=\left(X_{1}, X_{2}, \cdots, X_{N \times T}\right)^{\mathrm{T}}.
\end{array}
\end{equation*}

Therefore, the Hessian matrix of $h_{i}(\bm{\chi}_{i},\bm{\chi}_{-i})$ is positive semi-definite and $h_{i}(\bm{\chi}_{i},\bm{\chi}_{-i})$ is convex.
Moreover, since the utility function $U_{i}(.)$ is continuous and strictly concave in the strategy space, the payoff function Equ. (\ref{equ:u}) of each TE $i \enspace (\forall i \in \mathcal{N})$ is strictly concave. So the objective function in Equ. (\ref{equ:max_u1}) is concave. Hence, Equ. (\ref{equ:max_u}) is a convex optimization problem.
Meanwhile, since the constraints of Equ. (\ref{equ:max_u}) are inequalities or linear equations, the feasible domain is convex. Thus, the TEs' optimization problem is a convex programming problem.
Hence, the TE's game is a strictly concave $N$-person game. Since the demand profile sets are closed, bounded and convex, the existence of Nash equilibrium can be proved based on~\cite[Theorem 1]{Rosen1965Existence}. Analogously to~\cite[Theorem 3]{Rosen1965Existence}, for a concave $N$-person game, there exists a unique equilibrium solution. Therefore, the theorem is proved.
\end{proof}

In the Nash equilibrium, for any given ESs' bids, no TE can increase its payoff by a unilateral change on its strategy. In the next section, a task outsourcing algorithm is developed to determine the point for both ES and TE's games.

\section{Distributed Task Outsourcing Algorithm}
\label{sec:Algorithm}

In this section, we propose a distributed task outsourcing algorithm to demonstrate the interaction among TEs and ESs. Our method is referred as \textbf{\emph{DTOA}}. Let $g$ be the iteration number. 

Notations:

Let ${\chi}_{i,t}^{g}$ denote the demand profiles of TE $i$ in iteration $g$ at time slot $t$ and vector $\bm{\chi}_{i}^{g}$ denote the demand profile of TE $i$ for all time slots. The matrix $\bm{\chi}=\left(\bm{\chi}_{1}, \ldots,\bm{\chi}_{t}, \ldots,\bm{\chi}_{T}\right)^{T}$ denotes the demand profiles of all TEs in iteration $g$ for all time slots. Let matrix $\bm{\lambda}=\left(\bm{\lambda}_{1},\ldots,\bm{\lambda}_{t},\ldots,\bm{\lambda}_{T}\right)^{T}$ denote the bids of all ESs for all time slots. $L_{t}^{g}$ denotes the aggregate loads in iteration $g$ at time slot $t$. $p_{e}^{g}\left(\bm{\lambda}_{t}^{g}, L_{t}^{g}\right)$ denotes the computing service price in iteration $g$ at time slot $t$.

As shown in Fig. \ref{fig:Interaction}, the interaction between ESs and TEs can be modeled as a two-stage game. They interact with each other to determine optimal bids and demand profiles. The detailed process is depicted in Algorithm 1 and 2.
\begin{itemize}
	\item The ESs try to maximize their profits by determining their own bids according to optimization function Equ. (\ref{equ:max1}).
	\item The TEs will then adjust their demand profiles following optimization function Equ. (\ref{equ:max_u}).
\end{itemize}
\begin{figure}[htbp]
	\centering
	\includegraphics[width=1.0\linewidth]{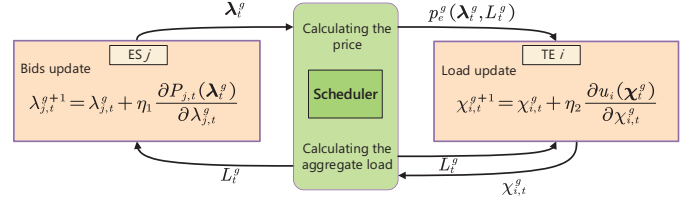}
	\caption{Interactions between the ESs, TEs and scheduler.} \label{fig:Interaction}
\end{figure}
\begin{algorithm}[htbp]
{
	{
	\caption{TE's game}
	\label{alg:u}
	\begin{algorithmic}[1]
		\STATE \textit{Initialization}: $g=0$.
		\STATE Randomly initialize TEs' demand profiles.
		\REPEAT
			\FOR{(each time slot $t \in \mathcal{T}$)}
				\STATE Receive $L_{t}^{g}$ from the scheduler.
				\STATE Update the bid $\bm{\lambda}_{t}^{g}$ by Algorithm \ref{alg:p}.
				\STATE Receive the updated $p_{e}^{g}\left(\bm{\lambda}_{t}^{g}, L_{t}^{g}\right)$ from the scheduler.
				\FOR{(each TE $i \in \mathcal{N}$)}
					\STATE $\chi_{i,t}^{g+1}=\left[\chi_{i,t}^{g}+\eta_{2} \frac{\partial u_{i}\left(\bm{\chi}_{t}^{g}\right)}{\partial \chi_{i,t}^{g}} \right]^{+}$.
				\ENDFOR
			\ENDFOR
			\STATE $g:=g+1$.			
		\UNTIL{$\left\| \bm{\chi}^{g} -\bm{\chi}^{g-1} \right\| < \epsilon$.}
	\end{algorithmic}
}
\par}
\end{algorithm}
\begin{algorithm}[htbp]
{
	{
	\renewcommand\baselinestretch{1.2}\selectfont 
	\caption{ES's game}
	\label{alg:p}
	\begin{algorithmic}[1]
		\REQUIRE Total load at time slot $t$: $L_{t}, t \in \mathcal{T}$ and $t$.
		\ENSURE Bids of all ESs at time slot t: $\bm{\lambda}_{t}$.
		\STATE \textit{Initialization}: Randomly initialize ESs' bid profiles for the first time.
		\STATE Receive $L_{t}$ from the scheduler.
		\FOR{(each ES $j \in \mathcal{M}$)}
			\STATE $\lambda_{j,t}^{g+1}=\left[\lambda_{j,t}^{g}+\eta_{1} \frac{\partial P_{j, t}\left(\bm{\lambda}_{t}^{g}\right)}{\partial \lambda_{j,t}^{g}}\right]^{+}$.
		\ENDFOR
		\RETURN $\bm{\lambda}_{t}$.
	\end{algorithmic}
	}
\par}
\end{algorithm}

The \textbf{\emph{DTOA}} can be described as follows. Firstly, the scheduler randomly initializes the TEs' demand profiles and ESs' bid profiles. Secondly, the TE $i \enspace (i \in \mathcal{N})$ sends the shiftable demand profile $\bm{\chi}_{i}^{g}$ to the broker and receives $L_{t}^{g}$ from it. Then, the ESs will receive a signal to update their bids based on the following iterative equation:
\begin{equation}
\label{equ:TEiteration} 
	\lambda_{j,t}^{g+1}=\left[\lambda_{j,t}^{g}+\eta_{1} \frac{\partial P_{j, t}\left(\bm{\lambda}_{t}^{g}\right)}{\partial \lambda_{j,t}^{g}}\right]^{+}, \quad \forall t \in \mathcal{T}.
\end{equation}
where $\eta_{1}$ is the step size.
$[\cdot]^{+}$ in Equ. (\ref{equ:TEiteration}) is the projection onto the feasible set defined by the constraints $\lambda_{j, t} \geq 0$.
It is noticed that the ES $j \enspace (j \in \mathcal{M})$ does not know other ESs' bids. In this aspect, the \textbf{\emph{DTOA}} can also preserve the privacy of participants. Thirdly, the computing service price $p_{e}^{g}\left(\bm{\lambda}_{t}^{g}, L_{t}^{g}\right)$ is updated by the scheduler according to Equ. (\ref{equ:p2}). The TEs will further be informed to update their shiftable demand profiles using a gradient boosting method:
\begin{equation}
\label{equ:ESiteration}
	\chi_{i,t}^{g+1}=\left[ \chi_{i,t}^{g}+\eta_{2} \frac{\partial u_{i}\left(\bm{\chi}_{t}^{g}\right)}{\partial \chi_{i,t}^{g}}\right]^{+}, \quad \forall t \in \mathcal{T}.
\end{equation}
$\eta_{2}$ is the step size.
$[\cdot]^{+}$ in Equ. (\ref{equ:ESiteration}) is the projection onto the feasible set defined by the constraints $\sum_{t \in \mathcal{T}} \chi_{i, t}=Q_{i}^{\text {total}}$ and $\chi_{i, t} \geq 0$.
It is worth remarking that Equ. (\ref{equ:max_u}) needs the updated price $p_{e}^{g}\left(\bm{\lambda}_{t}^{g}, L_{t}^{g}\right)$ and $L_{t}^{g}$ to determine $\left({\partial u_{i}\left(\bm{\chi}_{t}^{g}\right)}/{\partial \chi_{i,t}^{g}}\right)$. Besides, since $\left({\partial u_{i}\left(\bm{\chi}_{t}^{g}\right)}/{\partial \chi_{i,t}^{g}}\right)$ only depends on its own demand profile and the price and there is no need to know the demand profile of other TEs. Thus, this fact protects the privacy of the TEs. Finally, the stopping criterion of the algorithm is checked by the scheduler. If the relative change of shiftable demand profiles during two consecutive iterations is lower than the value $\epsilon$, the iterations can be stopped. Otherwise, the TEs will continue computing their demand profiles based on the newly updated price and bids.

The optimization problems Equ. (\ref{equ:max1}) and Equ. (\ref{equ:max_u}) will converge to the optimal point by the projected gradient method. In the end, the algorithm will converge.
In the equilibrium, the ESs are playing their equilibrium strategies according to TEs' tasks strategies, and the TEs also choose their equilibrium strategies based on ESs' submitted bids.
Thus when the Nash equilibrium is reached, none of the ESs and TEs improve their profit.
\section{Performance Evaluation}
\label{sec:Performance Evaluation}

\subsection{Simulation experiment} 

In this section, we present a simulation experiment to validate our theoretical analysis.
We assume a MEC resource exchange market has 10 ESs and 1000 TEs, which are willing to participate in the \textbf{\emph{DTOA}} scheme.
There are 24 time slots.
The relevant parameters of the model are shown in Table \ref{table:02}.
The base demand $r_{i,t}$ of each TE at each time slot is randomly selected from [9660, 37065].
The shiftable demand refers to real-time, non-shiftable tasks, which reflects the changes in the total demand of all TEs at different time slots.
Since most loads are running in real-time pattern, it is plausible to assume relatively low shiftable loads for TEs.
The shiftable demand ${\chi}_{i,t}$ of each TE is assumed to be chosen randomly from 10\% to 12\% of its base demand.
And the total demand is the sum of the base demand and shiftable demand.
Considering the generation cost function $c\left(f_{j, t}\right)=a_{j,2} f_{j, t}^{2}+a_{j,1} f_{j, t}+a_{j,0}$ for each ES $j\enspace (j \in \mathcal{M})$, we assume that $a_{j,2}$ is randomly generated in the interval [4.76e-6, 4.76e-5], $a_{j,1}=0.001$ and $a_{j,0}=0.001$.
The initial values of $\eta_{1}$ and $\eta_{2}$ are set as 0.05 and 0.01 respectively.
In order to find the optimal solution, the step size of next iteration will be a little less than the previous one, namely $\eta_{1}$=$\eta_{1}$*0.985 and $\eta_{2}$=$\eta_{2}$*0.98.
The initial bids $\lambda_{j,t}$ of ESs are all set as 20000.
The $\alpha_{i,t}$ is set as 0.5 and $\omega_{i,t}$ is randomly selected from interval [0.8, 1.0].
Also, the $\epsilon$ is set equal to 0.3.

\begin{table}\centering
	\caption{System Parameters}
	{
	  \begin{tabular}{ll}
		\hline \textbf{System parameters} & \textbf{Value(Fixed)-[Varied range]} \\
		\hline 
		 Base demand $r_{i,t}$ & [9660, 37065] \\
		 Shiftable demand ${\chi}_{i,t}$ & [10\%,12\%]*Base demand \\
		 $a_{j,2}$ & [4.76e-6, 4.76e-5] \\
		 $a_{j,1}$ & (0.001) \\
		 $a_{j,0}$ & (0.001) \\
		 step size $\eta_{1}$ & (0.05), $\eta_{1}$=$\eta_{1}$*0.985 \\
		 step size $\eta_{2}$ & (0.001), $\eta_{2}$=$\eta_{2}$*0.98 \\
		 ES's bid $\lambda_{j,t}$ & (20000) \\
		 $\alpha_{i,t}$ & (0.5) \\		 
		 $\omega_{i,t}$ & [0.8,1.0] \\
		 $\epsilon$ & (0.3)\\
		\hline	
	  \end{tabular}
	}
  \label{table:02}
\end{table}


\begin{figure}[htbp]
	\centering
	\includegraphics[width=1.0\linewidth]{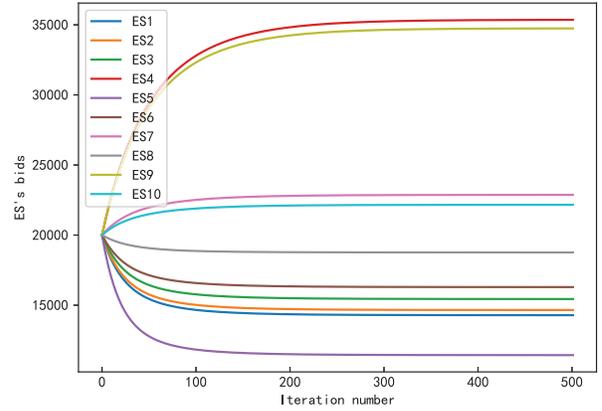}
	\caption{Convergence of ES 1-10's bids at time slot 5.} \label{fig:Provider_Bid}
\end{figure}
\begin{figure}[htbp]
	\centering
	\includegraphics[width=1.0\linewidth]{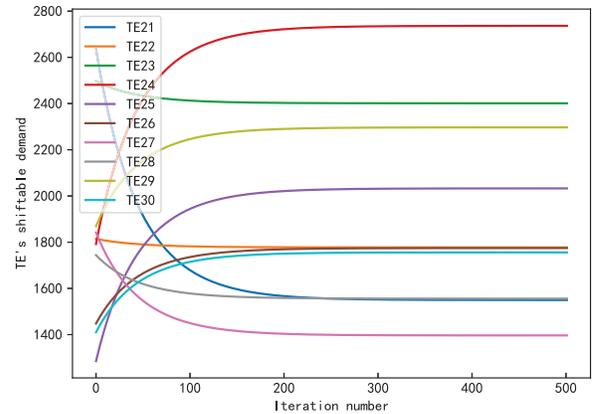}
	\caption{Convergence of shiftable loads for TE 21-30 at time slot 5.} \label{fig:User_Shift}
\end{figure}

\subsection{Algorithm Convergence}
The performance of our proposed \textbf{\emph{DTOA}} is evaluated in terms of its convergence.
Fig. \ref{fig:Provider_Bid} and Fig. \ref{fig:User_Shift} show the convergence of ESs' bids and TEs' shiftable loads at time slot 5.
From Fig. \ref{fig:User_Shift}, these ten TEs (TEs 21-30) are randomly selected from 1000 TEs.
The speed of convergence to the equilibrium point depends on the step sizes and the stopping criterion $\epsilon$.
As the number of iterations increases, the bids and the shiftable load demands start from the initial values and they gradually converge to stable values.
In our experiment, the algorithm converges after around 248 iterations.
Hence, the proposed \textbf{\emph{DTOA}} is efficient and verifies the theoretical proof presented above.

To demonstrate the computational complexity of the algorithm, we evaluate the running time of the algorithm for different number of TEs and ESs.
As shown in Fig. \ref{fig:RunningTime}, the running time of the algorithm increases linearly with the number of TEs $N$ and it is almost independent of $M$.
This is because that by increasing the number of TEs and ESs, the number of updates for TEs and ESs will increase proportional to $N$ and $M$, respectively.
The update process for TEs takes more time comparing with the updates for ESs since the TEs need to consider load shifting during $T$ time slots (the projected gradient), which make the update process more complex.
From Fig. \ref{fig:RunningTime}, the running time of the algorithm is acceptable even for large number of TEs.
So it can be concluded that the algorithm is efficient and can be implemented in scenarios with large number of TEs.
\begin{figure}[htbp]
	\centering
	\includegraphics[width=1.0\linewidth]{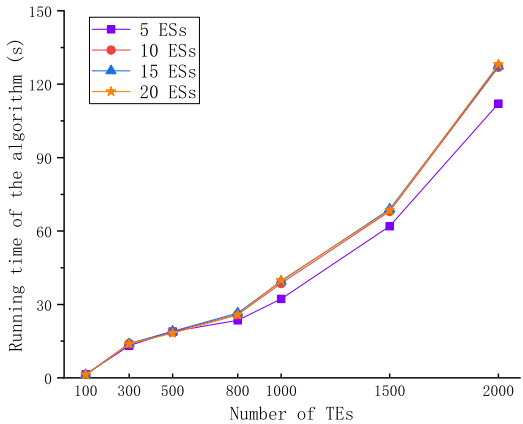}
	\caption{Running time of algorithm for different number of TEs and ESs}\label{fig:RunningTime}
\end{figure}

\subsection{Economic effects of algorithm}
By participating in the \textbf{\emph{DTOA}} scheme, the payout (see Equ. (\ref{equ:payout})) represents TE's expenditure on purchasing computing resources.
Fig. \ref{fig:User_Payout} shows the daily total payout for TE 1 to TE 30 before algorithm and after algorithm.
Compared with before algorithm, the total payout of each TE after algorithm is reduced.
We can see that TEs can save around 5\% of they payout by participating in \textbf{\emph{DTOA}} scheme.
The vertical axis of Fig. \ref{fig:User_Payout} shows the payouts of TEs are so huge, and even a 5\% savings reduces a large expenditure.
Furthermore, as shown in Fig. \ref{fig:User_Utility}, the daily total payoff of each TE after using the algorithm has increased than before algorithm.
The payoff (see Equ. (\ref{equ:u})) is the utility of the calculation tasks minus the payout.
Although the green bar is only a little more than the red bar chart, the payoff has also increased a lot because its magnitude is large and arrives $10^{9}$.

\begin{figure*}[htbp]
	\centering
	\includegraphics[width=1.0\linewidth]{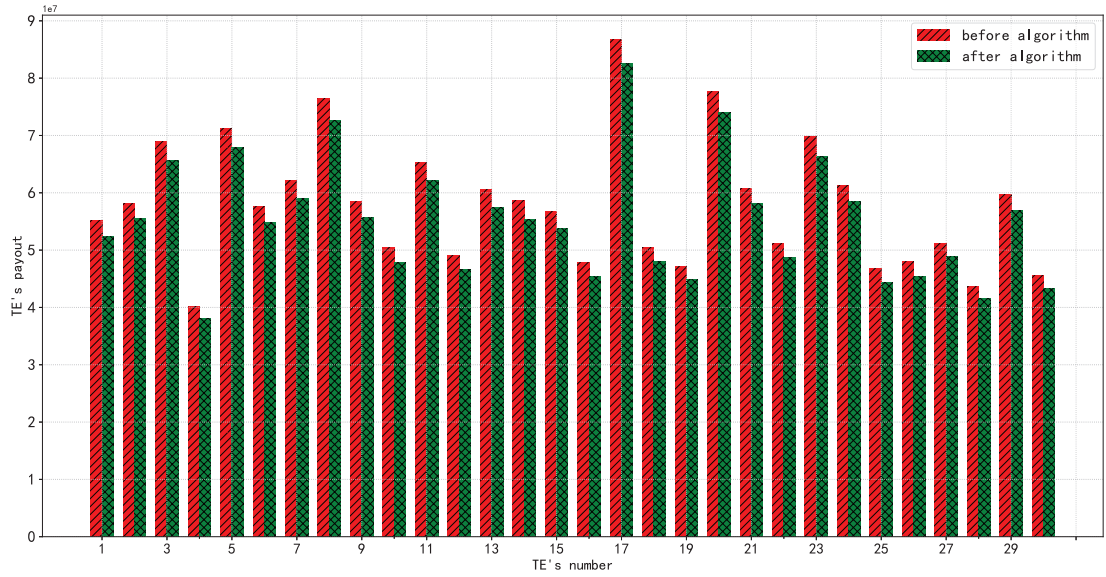}
	\caption{Daily total payout for TE 1 to TE 30.} \label{fig:User_Payout}
\end{figure*}
\begin{figure*}[htbp]
	\centering
	\includegraphics[width=1.0\linewidth]{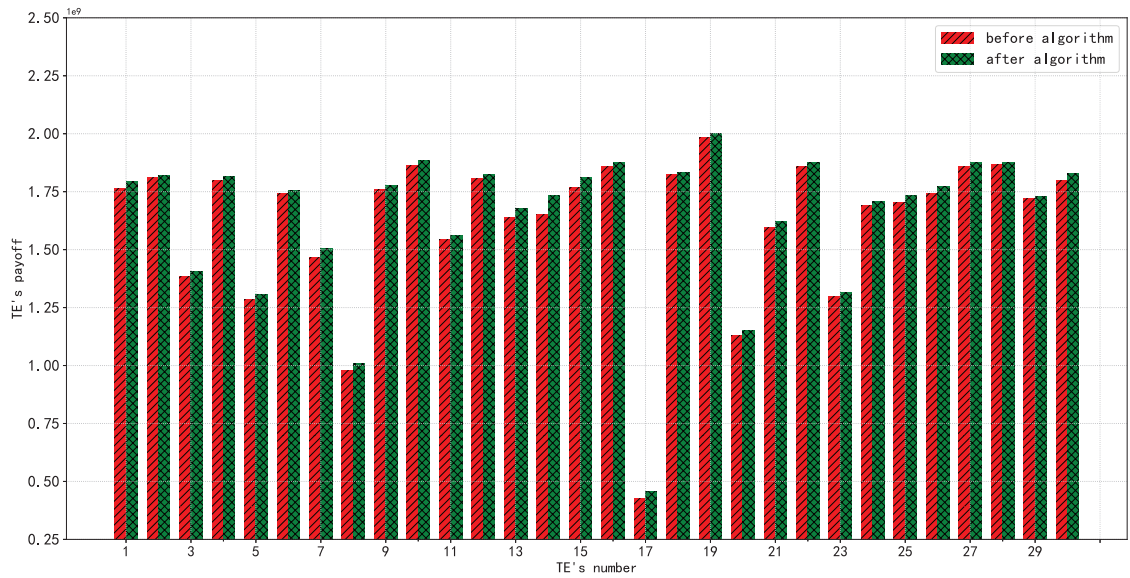}
	\caption{Daily total payoff for TE 1 to TE 30.} \label{fig:User_Utility}
\end{figure*}

Fig. \ref{fig:P_Profit} displays the total profit of ESs 1-10 before and after algorithm.
The total profit of ES is the sum of the profit of all time slots.
From Fig. \ref{fig:P_Profit}, we can see that the total profit of each ES increases after applying \textbf{\emph{DTOA}} because the aggregate load profile becomes smoother; and hence, the ESs' generation cost decreases.
Besides, the suppliers aim to submit optimal bids that maximize their profits in each time slot.
The results of the algorithm are in line with expectations, which shows the supplier side's individual rationality of our proposed method.
\begin{figure}[htbp]
	\centering
	\includegraphics[width=1.0\linewidth]{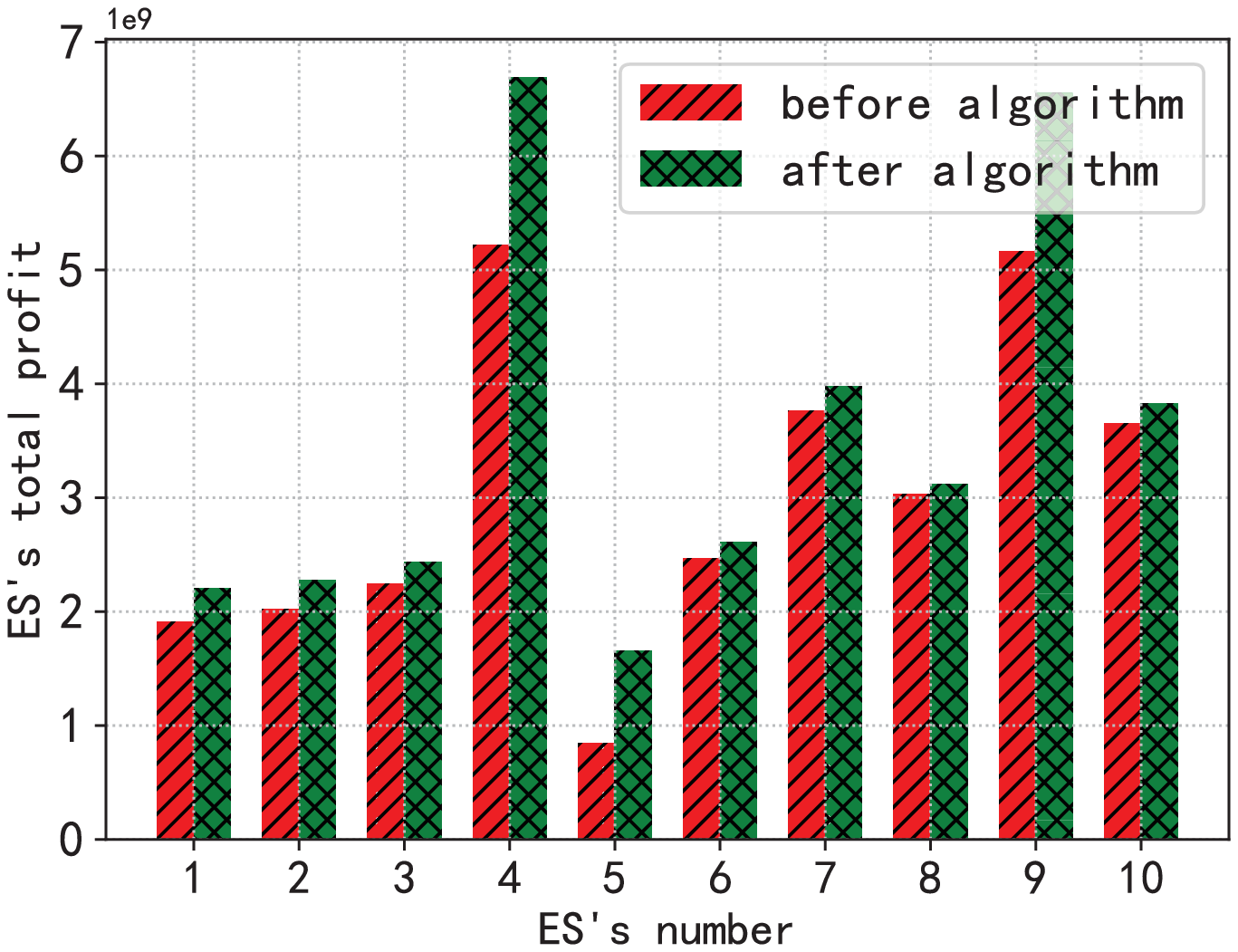}
	\caption{Daily total profit for ES 1 to ES 10.} \label{fig:P_Profit}
\end{figure}

\subsection{Peak-reducing effect of algorithm}

\begin{figure}[htbp]
	\centering
	\includegraphics[width=1.0\linewidth]{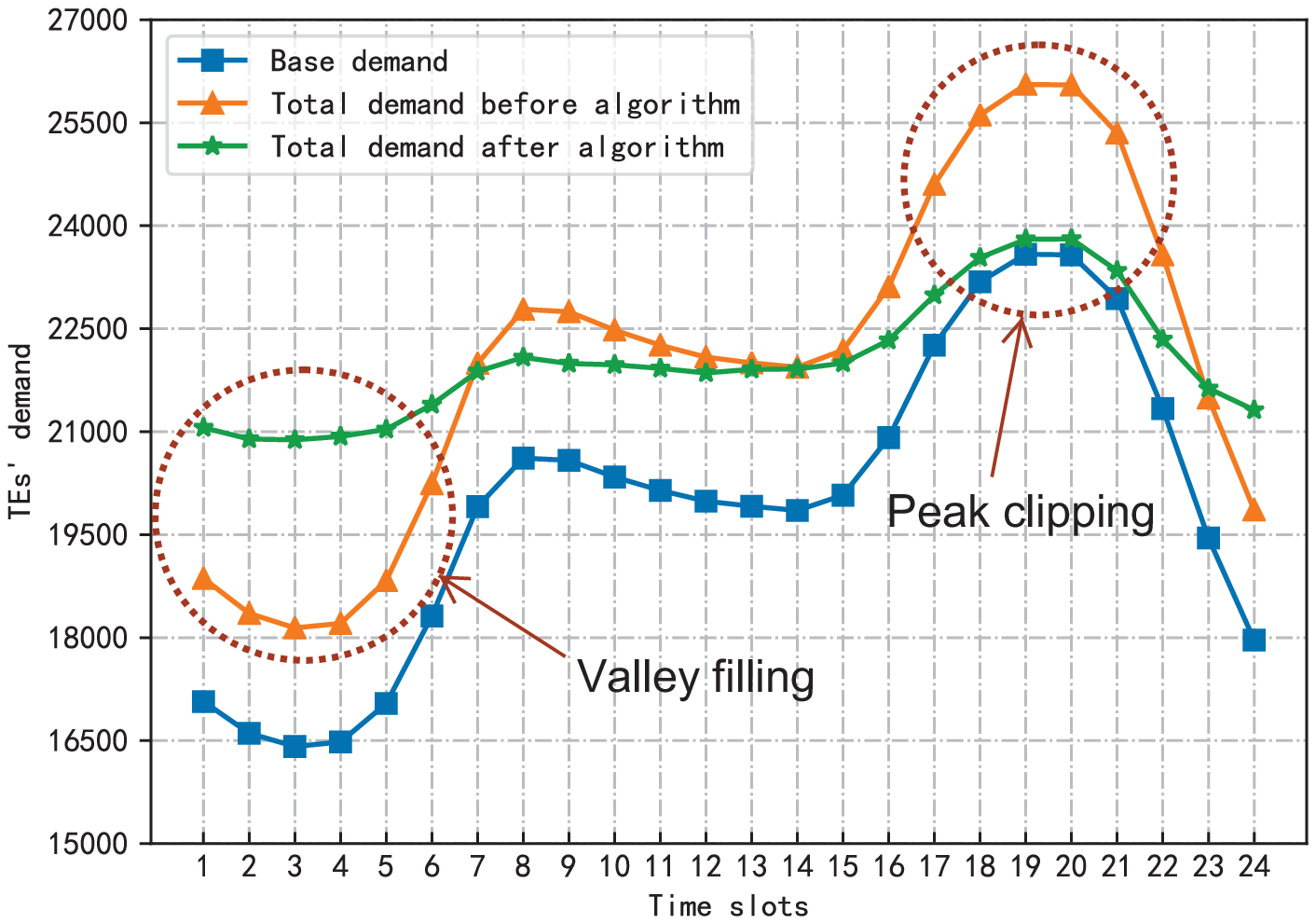}
	\caption{Base and total demand before and after algorithm. The peak shaving is achieved by using \textbf{\emph{DTOA}} (the dashed circle).} \label{fig:Demand}
\end{figure}

For simulations, the initial state of TEs' demand is assumed to be load profile before algorithm.
The aggregate load profile becomes smoother after the \textbf{\emph{DTOA}}.
The dashed circle shows the fluctuation of the demand including valley filling and peak clipping.
Normally, the peak load demand is 26200, while the peak load demand decreases to 23800 in the case of the \textbf{\emph{DTOA}}.
Therefore, the peak load demands are shifted from peak to off-peak time slots. 
Furthermore, the load demand for each TE is shifted to time slots with higher $w_{j,t}$, which brings a higher payoff to the TEs.
This demonstrates that the proposed \textbf{\emph{DTOA}} performs satisfactorily in reducing the peak load demand.

\begin{figure}[htbp]
	\centering
	\includegraphics[width=1.0\linewidth]{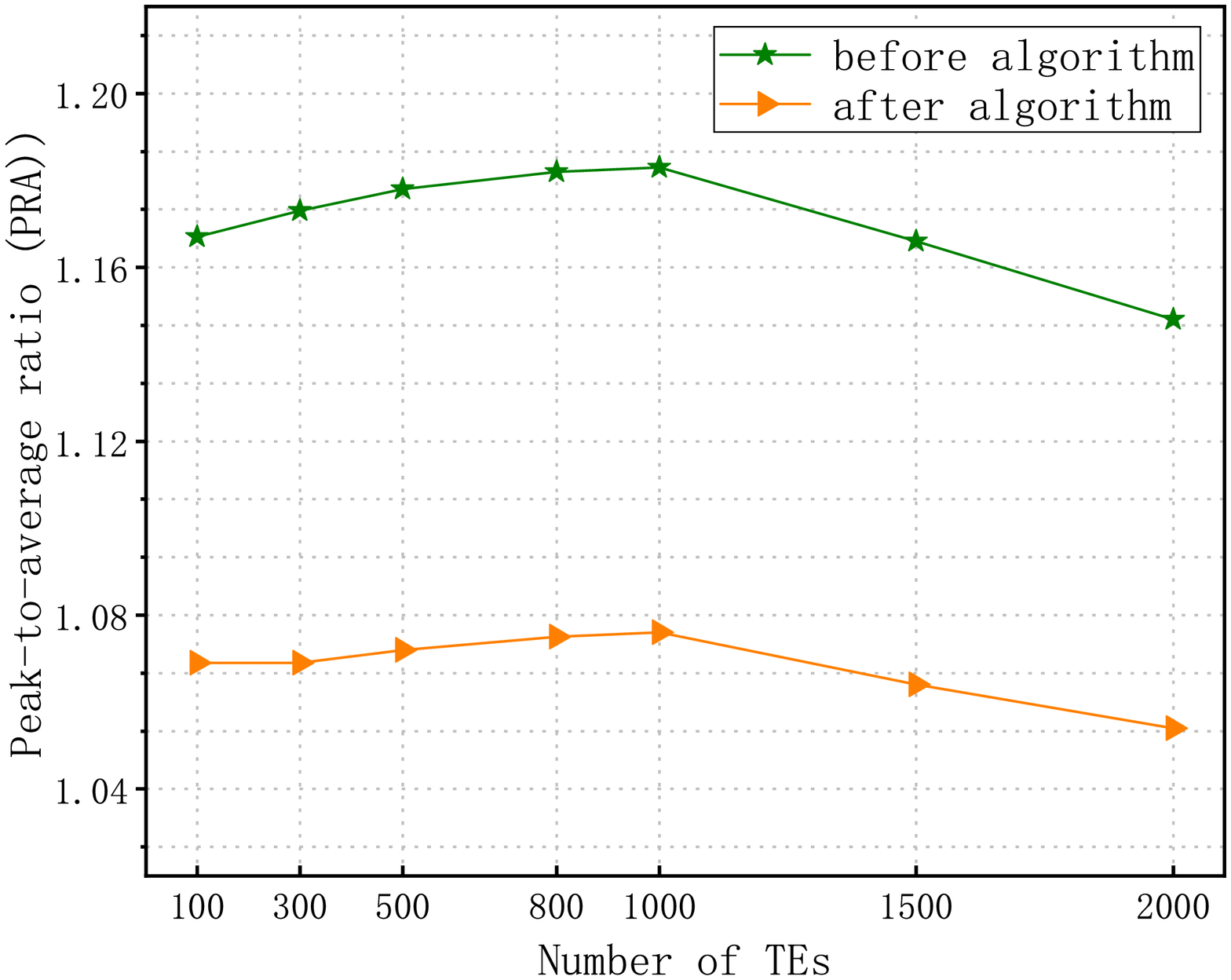}
	\caption{Peak-to-average ratio with and without task scheduling}\label{fig:PAR}
\end{figure}
Fig. \ref{fig:PAR} shows the PAR (peak-to-average ratio) index with and without task scheduling in 24 time slots and for different number of TEs.
Before algorithm, since there are high peak load and low average load, PAR index is high.
By applying \textbf{\emph{DTOA}}, the peak clipping and valley filling are achieved and the PAR index is low even for high number of TEs.
This demonstrates that the proposed task scheduling method can shift the shiftable loads from peak periods to off-peak periods effectively.

\subsection{Influence of parameters on iteration numbers} 

In this section, we discuss the influence of some parameters on the convergence speed of the algorithm.
The convergence speed of the algorithm is reflected in the round of algorithm updates (iteration numbers).
The smaller the iteration numbers, the faster the algorithm converges.
The bigger the iteration numbers, the slower the algorithm converges.

Fig. \ref{fig:theta} shows the influence of the parameter $\epsilon$ on iteration numbers.
$\epsilon$ is the stopping criterion of the algorithm.
As can be seen, the smaller the parameter $\epsilon$, the more iterations and the slower the algorithm convergence.
This fact shows that the stricter of the stopping criterion, the more times the algorithm needs to be updated.

In Fig. \ref{fig:r2}, the influence of the parameter $\eta_{2}$ on iteration numbers are shown.
Since the parameter $\eta_{2}$ will change in every round, as shown in Table \ref{table:02}, we set different initial value of parameter $\eta_{2}$ to show its impact on iteration numbers.
As can be seen, when other parameters are fixed, with the initial value of parameter $\eta_{2}$ becomes larger, the number of iterations also increases.
In summary, the speed of the algorithm convergence is related to the setting of some parameters.

\begin{figure}[htbp]
	\centering
	\includegraphics[width=1.0\linewidth]{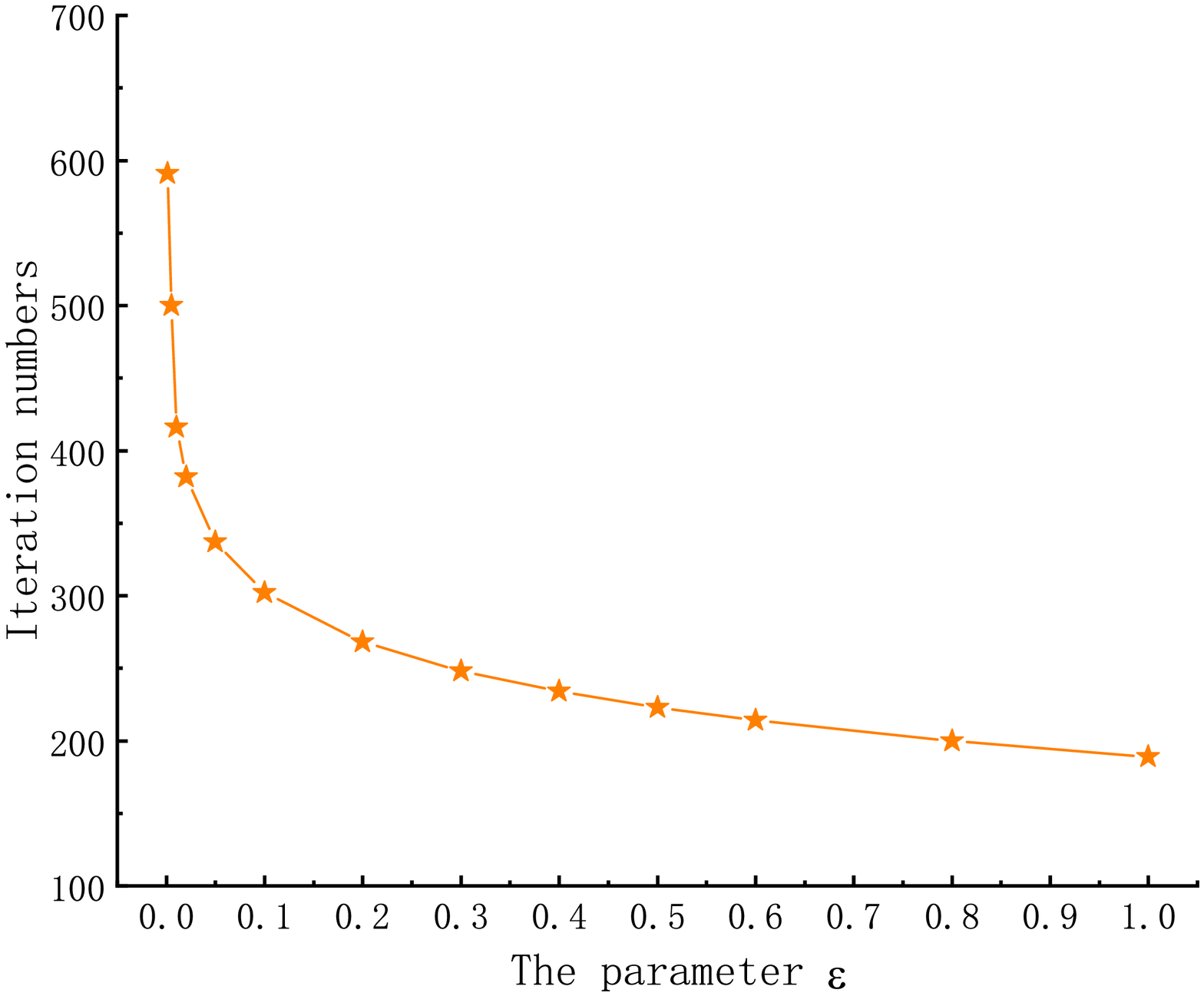}
	\caption{The influence of the parameter $\epsilon$ on iteration numbers.}\label{fig:theta}
\end{figure}
\begin{figure}[htbp]
	\centering
	\includegraphics[width=1.0\linewidth]{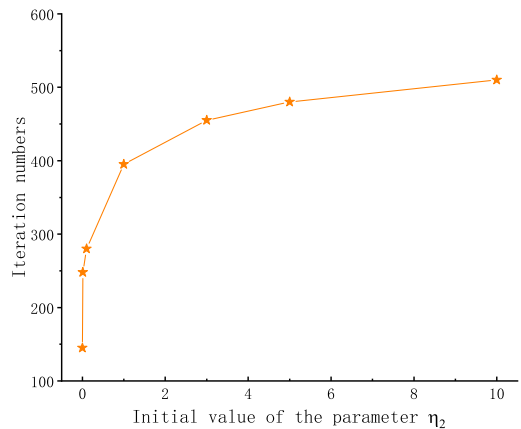}
	\caption{The influence of the parameter $\eta_{2}$ on iteration numbers.}\label{fig:r2}
\end{figure}
\section{Conclusion}
\label{sec:Conclusion}

In this paper, we analyze a practical resources transaction market in a MEC network, where multiple different ESs offering the optional computing service to TEs.
Since the resources of each ES are limited, the dynamic demand of its TEs may not be met during spikes in demands.
To overcome the bottleneck of resource limitation, task outsourcing has been regarded as an effective paradigm by accommodating as many on-demand tasks as possible.
We focus on the task outsourcing problem among multiple ESs and multiple TEs.
A bidding mechanism is utilized to describe the serving relationship between ESs and TEs, where the two parties are assigned as sellers and buyers.
The computing resources of ESs are regarded as commodities.

Simulations results demonstrate that the algorithm increases the ESs' profit and reduces the peak load by shifting the load demand to off-peak periods. Meanwhile, the TEs' payoff are also increased by participating in game process. As for future research, we will focus on the computing offloading of ESs in a three-tier IoT MEC networks.

\ifCLASSOPTIONcompsoc
  \section*{Acknowledgments}
\else
  \section*{Acknowledgment}
\fi
This work is supported by National Natural Science Foundation of China (No.61872129).
\ifCLASSOPTIONcaptionsoff
  \newpage
\fi

\bibliography{reference}{}

\begin{thebibliography}{10}
\providecommand{\url}[1]{#1}
\csname url@samestyle\endcsname
\providecommand{\newblock}{\relax}
\providecommand{\bibinfo}[2]{#2}
\providecommand{\BIBentrySTDinterwordspacing}{\spaceskip=0pt\relax}
\providecommand{\BIBentryALTinterwordstretchfactor}{4}
\providecommand{\BIBentryALTinterwordspacing}{\spaceskip=\fontdimen2\font plus
\BIBentryALTinterwordstretchfactor\fontdimen3\font minus
  \fontdimen4\font\relax}
\providecommand{\BIBforeignlanguage}[2]{{%
\expandafter\ifx\csname l@#1\endcsname\relax
\typeout{** WARNING: IEEEtran.bst: No hyphenation pattern has been}%
\typeout{** loaded for the language `#1'. Using the pattern for}%
\typeout{** the default language instead.}%
\else
\language=\csname l@#1\endcsname
\fi
#2}}
\providecommand{\BIBdecl}{\relax}
\BIBdecl

\bibitem{Shuiguang2020Optimal}
S.~Deng, Z.~Xiang, J.~Taheri, K.~A. Mohammad, and S.~Dustdar, ``Optimal
  application deployment in resource constrained distributed edges,''
  \emph{IEEE Transactions on Mobile Computing}, vol.~PP, no.~99, pp. 1--1,
  2020.

\bibitem{Cosmin2019Decentralized}
C.~Avasalcai, C.~Tsigkanos, and S.~Dustdar, ``Decentralized resource auctioning
  for latency-sensitive edge computing,'' in \emph{2019 IEEE International
  Conference on Edge Computing (EDGE)}, 2019.

\bibitem{Nafiseh2019QoS}
N.~Sharghivand, F.~Derakhshan, and L.~Mashayekhy, ``Qos-aware matching of edge
  computing services to internet of things,'' in \emph{2018 IEEE 37th
  International Performance Computing and Communications Conference (IPCCC)},
  2019.

\bibitem{lyu2018selective}
X.~Lyu, H.~Tian, L.~Jiang, A.~Vinel, S.~Maharjan, S.~Gjessing, and Y.~Zhang,
  ``Selective offloading in mobile edge computing for the green internet of
  things,'' \emph{IEEE Network}, vol.~32, no.~1, pp. 54--60, 2018.

\bibitem{zhao2015cooperative}
T.~Zhao, S.~Zhou, X.~Guo, Y.~Zhao, and Z.~Niu, ``A cooperative scheduling
  scheme of local cloud and internet cloud for delay-aware mobile cloud
  computing,'' in \emph{2015 IEEE Globecom Workshops (GC Wkshps)}.\hskip 1em
  plus 0.5em minus 0.4em\relax IEEE, 2015, pp. 1--6.

\bibitem{Grosu2005Noncooperative}
D.~Grosu and A.~T. Chronopoulos, ``Noncooperative load balancing in distributed
  systems,'' \emph{Journal of Parallel \& Distributed Computing}, vol.~65,
  no.~9, pp. 1022--1034, 2005.

\bibitem{jiang2016a}
Y.~Jiang, ``A survey of task allocation and load balancing in distributed
  systems,'' \emph{IEEE Transactions on Parallel and Distributed Systems},
  vol.~27, no.~2, pp. 585--599, 2016.

\bibitem{yan2018game}
S.~Yan, M.~Peng, and X.~Cao, ``A game theory approach for joint access
  selection and resource allocation in uav assisted iot communication
  networks,'' \emph{IEEE Internet of Things Journal}, vol.~6, no.~2, pp.
  1663--1674, 2018.

\bibitem{you2016energy}
C.~You, K.~Huang, H.~Chae, and B.-H. Kim, ``Energy-efficient resource
  allocation for mobile-edge computation offloading,'' \emph{IEEE Transactions
  on Wireless Communications}, vol.~16, no.~3, pp. 1397--1411, 2016.

\bibitem{li2019radio}
C.~Li, W.~Chen, J.~Tang, and Y.~Luo, ``Radio and computing resource allocation
  with energy harvesting devices in mobile edge computing environment,''
  \emph{Computer Communications}, 2019.

\bibitem{zhao2019computation}
J.~Zhao, Q.~Li, Y.~Gong, and K.~Zhang, ``Computation offloading and resource
  allocation for cloud assisted mobile edge computing in vehicular networks,''
  \emph{IEEE Transactions on Vehicular Technology}, vol.~68, no.~8, pp.
  7944--7956, 2019.

\bibitem{gao2019optimal}
Y.~Gao, Y.~Cui, X.~Wang, and Z.~Liu, ``Optimal resource allocation for scalable
  mobile edge computing,'' \emph{IEEE Communications Letters}, 2019.

\bibitem{avasalcai2019latency}
C.~Avasalcai and S.~Dustdar, ``Latency-aware distributed resource provisioning
  for deploying iot applications at the edge of the network,'' in \emph{Future
  of Information and Communication Conference}.\hskip 1em plus 0.5em minus
  0.4em\relax Springer, 2019, pp. 377--391.

\bibitem{zhang2017combinational}
H.~Zhang, F.~Guo, H.~Ji, and C.~Zhu, ``Combinational auction-based service
  provider selection in mobile edge computing networks,'' \emph{IEEE Access},
  vol.~5, pp. 13\,455--13\,464, 2017.

\bibitem{gu2019task}
B.~Gu and Z.~Zhou, ``Task offloading in vehicular mobile edge computing: A
  matching-theoretic framework,'' \emph{IEEE Vehicular Technology Magazine},
  vol.~14, no.~3, pp. 100--106, 2019.

\bibitem{zhang2017computing}
H.~Zhang, Y.~Xiao, S.~Bu, D.~Niyato, F.~R. Yu, and Z.~Han, ``Computing resource
  allocation in three-tier iot fog networks: A joint optimization approach
  combining stackelberg game and matching,'' \emph{IEEE Internet of Things
  Journal}, vol.~4, no.~5, pp. 1204--1215, 2017.

\bibitem{liu2017price}
M.~Liu and Y.~Liu, ``Price-based distributed offloading for mobile-edge
  computing with computation capacity constraints,'' \emph{IEEE Wireless
  Communications Letters}, vol.~7, no.~3, pp. 420--423, 2017.

\bibitem{tang2019jointly}
Q.~Tang, R.~Xie, T.~Huang, and Y.~Liu, ``Jointly caching and computation
  resource allocation for mobile edge networks,'' \emph{IET Networks}, vol.~8,
  no.~5, pp. 329--338, 2019.

\bibitem{jalali2015demand}
M.~M. Jalali and A.~Kazemi, ``Demand side management in a smart grid with
  multiple electricity suppliers,'' \emph{Energy}, vol.~81, pp. 766--776, 2015.

\bibitem{samadi2010optimal}
P.~Samadi, A.-H. Mohsenian-Rad, R.~Schober, V.~W. Wong, and J.~Jatskevich,
  ``Optimal real-time pricing algorithm based on utility maximization for smart
  grid,'' in \emph{2010 First IEEE International Conference on Smart Grid
  Communications}.\hskip 1em plus 0.5em minus 0.4em\relax IEEE, 2010, pp.
  415--420.

\bibitem{kamyab2015demand}
F.~Kamyab, M.~Amini, S.~Sheykhha, M.~Hasanpour, and M.~M. Jalali, ``Demand
  response program in smart grid using supply function bidding mechanism,''
  \emph{IEEE Transactions on Smart Grid}, vol.~7, no.~3, pp. 1277--1284, 2015.

\bibitem{baldick2004theory}
R.~Baldick, R.~Grant, and E.~Kahn, ``Theory and application of linear supply
  function equilibrium in electricity markets,'' \emph{Journal of regulatory
  economics}, vol.~25, no.~2, pp. 143--167, 2004.

\bibitem{chen2006coevolutionary}
H.~Chen, K.~Wong, C.~Chung, and D.~Nguyen, ``A coevolutionary approach to
  analyzing supply function equilibrium model,'' \emph{IEEE Transactions on
  Power Systems}, vol.~21, no.~3, pp. 1019--1028, 2006.

\bibitem{johari2006parameterized}
R.~Johari and J.~N. Tsitsiklis, ``Parameterized supply function bidding:
  equilibrium and welfare,'' \emph{Mathematics of Operations Research}, 2006.

\bibitem{Rosen1965Existence}
J.~B. Rosen, ``Existence and uniqueness of equilibrium points for concave
  n-person games,'' \emph{Econometrica}, vol.~33, no.~3, pp. 520--534, 1965.

\end{thebibliography}
\bibliographystyle{IEEEtran}
\begin{IEEEbiography}[{\includegraphics[width=2.0in,height=1.25in,clip,keepaspectratio]{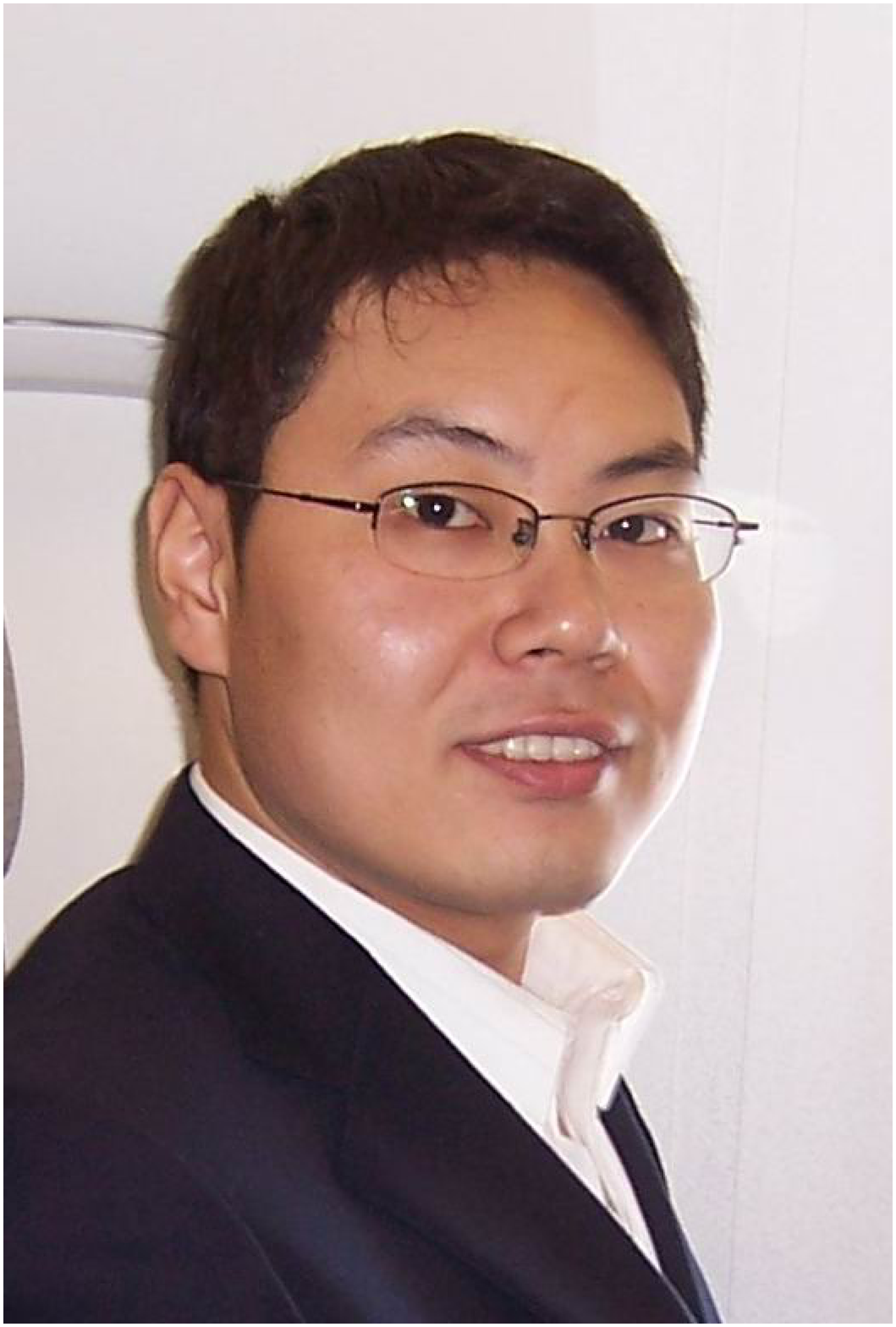}}]{Zheng Xiao}
	received the Ph.D. in Computer Science from Fudan University, China, in 2009, and B.S. in Communication Engineering from Hunan University in 2003. He is currently an associate professor in College of Information Science and Engineering of Hunan University. His major research interests include distributed artificial intelligence, high performance computing, parallel and distributed systems, intelligent information processing and collaborative optimization. He has published over 30 journal articles and conference papers. He is currently an associate editor of IEEE Access. He is a member of IEEE and CCF.
\end{IEEEbiography}

\begin{IEEEbiography}[{\includegraphics[width=1in,height=1.25in,clip,keepaspectratio]{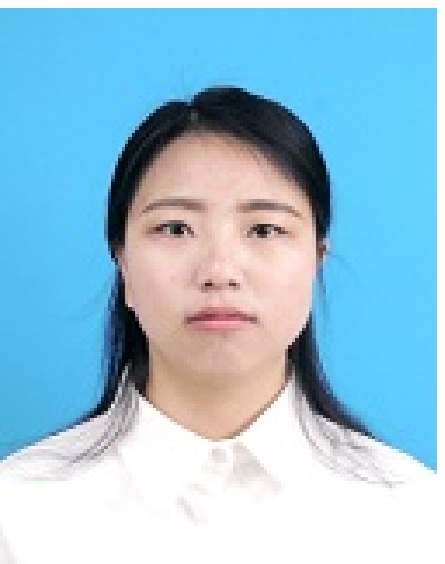}}]{Dan He}
	received the B.S. degree in computer science and technology from Jiangxi Normal University, Nanchang, China, in 2018. She is currently working toward the M.S. degree at the College of Information Science and Engineering, Hunan University, Changsha, China. Her research interests focus on high performance computing, modeling and resource scheduling in cloud computing systems, big data pricing and game theory.
\end{IEEEbiography}

\begin{IEEEbiography}[{\includegraphics[width=1in,height=1.25in,clip,keepaspectratio]{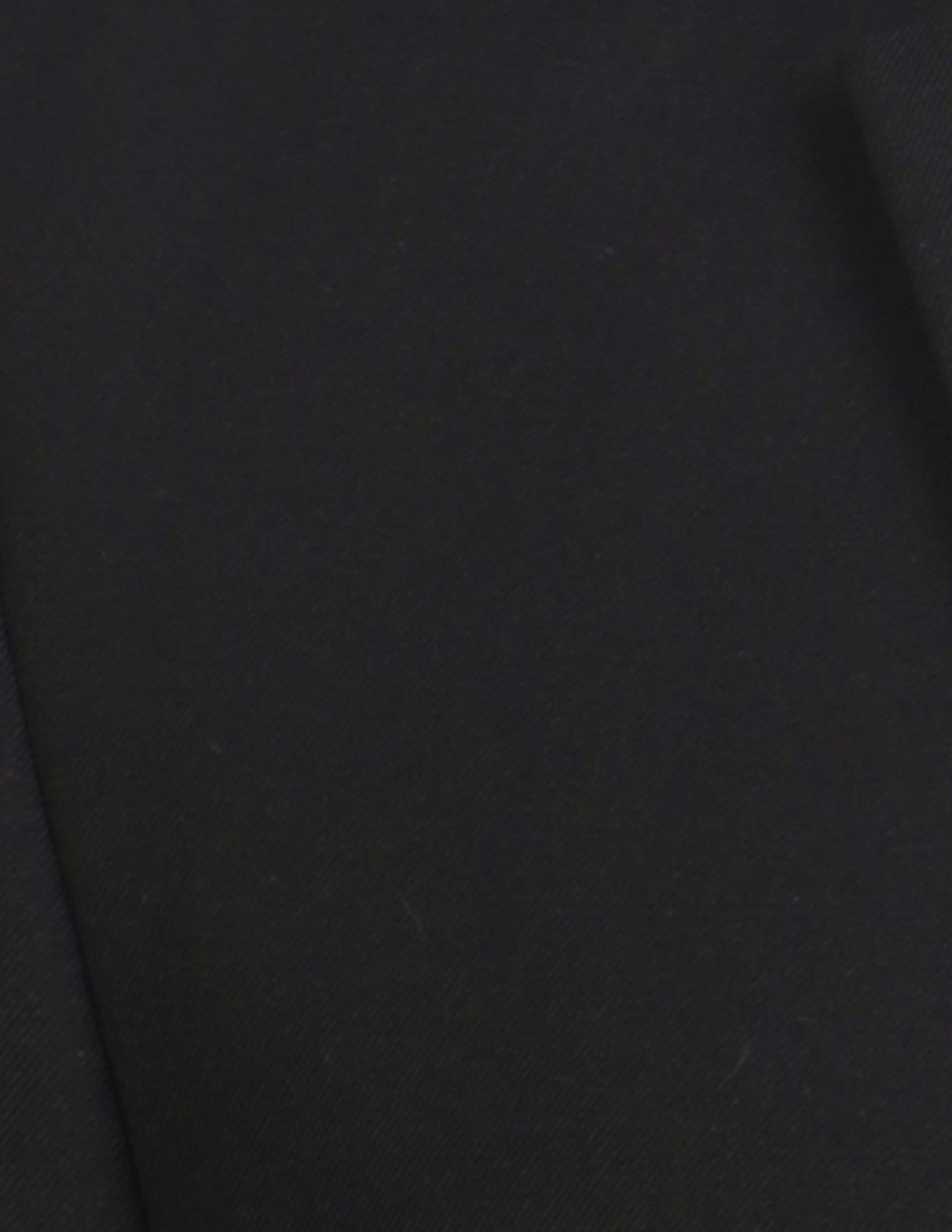}}]{Yu Chen}
    is currently working toward the M.S. degree at the Collage of Information Science and Engineering, Hunan University, China.
	His research interests focus on machine learning and natural language processing.
\end{IEEEbiography}

\begin{IEEEbiography}[{\includegraphics[width=1in,height=1.25in,clip,keepaspectratio]{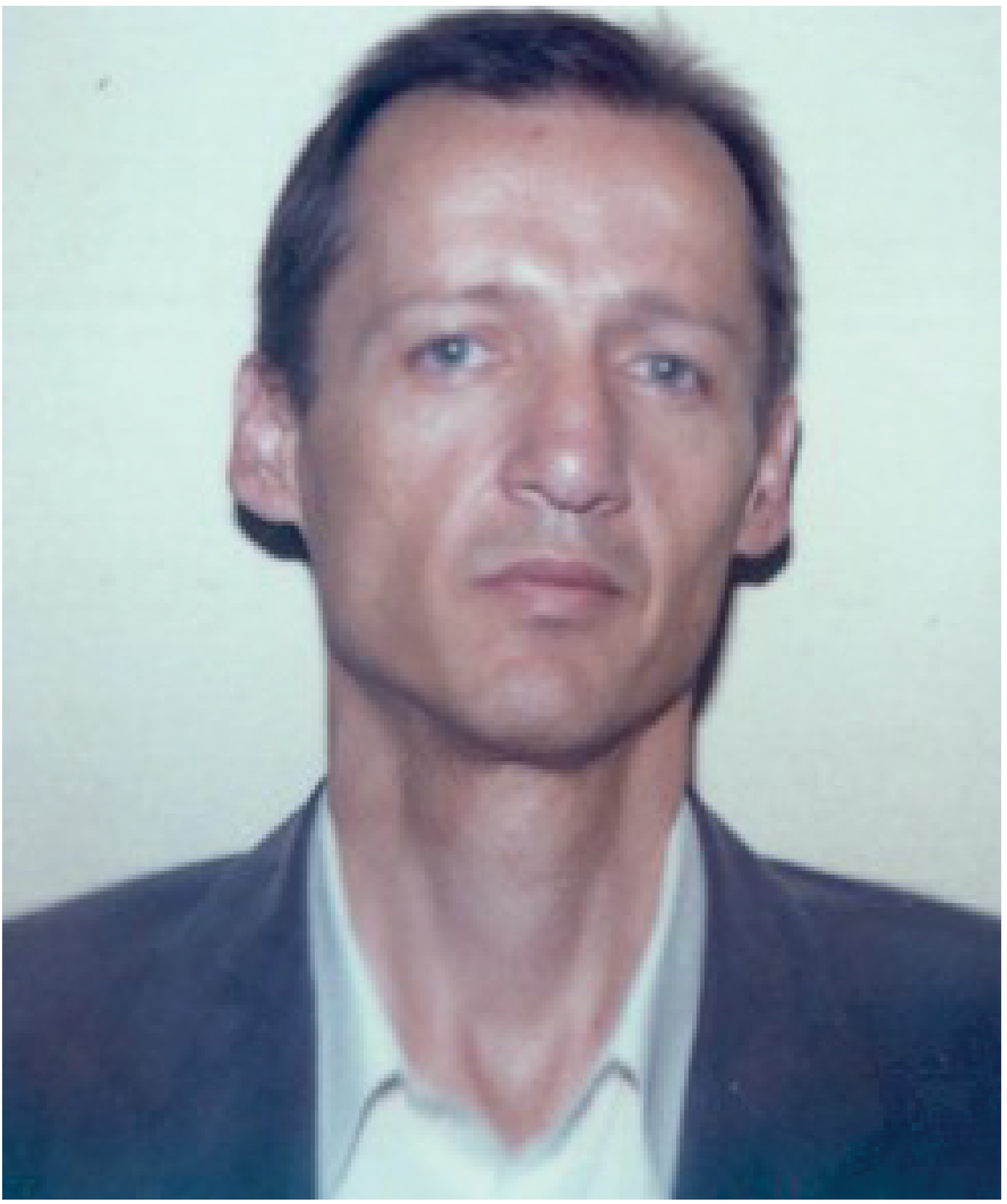}}]{Anthony Theodore Chronopoulos}
	obtained a Ph.D. in Computer Science from the University of Illinois at Urbana-Champaign in 1987. He is a full Professor at the Department of Computer Science, University of Texas, San Antonio, USA and a visiting professor, Department of Computer Engineering \& Informatics, University of Patras, Greece. He is the author of 83 journal and 73 peer-reviewed conference proceedings publications in the areas of Distributed and Parallel Computing, Grid and Cloud Computing, Scientific Computing, Computer  Networks, Computational Intelligence. He is a Fellow of the Institution of Engineering and Technology (FIET), ACM Senior member, \emph{IEEE} Senior member.
\end{IEEEbiography}

\begin{IEEEbiography}[{\includegraphics[width=1in,height=1.25in,clip,keepaspectratio]{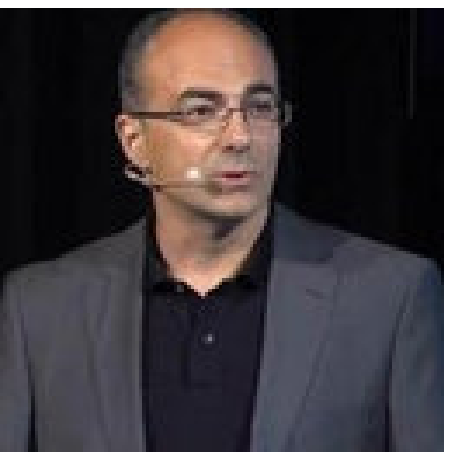}}]{Schahram Dustdar}
    is Full Professor of Computer Science (Informatics) with a focus on Internet Technologies heading the Distributed Systems Group at the TU Wien. He is Chairman of the Informatics Section of the Academia Europaea (since December 9, 2016).
    He is elevated to IEEE Fellow (since January 2016).
    From 2004-2010 he was Honorary Professor of Information Systems at the Department of Computing Science at the University of Groningen (RuG), The Netherlands.
    From December 2016 until January 2017 he was a Visiting Professor at the University of Sevilla, Spain and from January until June 2017 he was a Visiting Professor at UC Berkeley, USA.
    He is a member of the IEEE Conference Activities Committee (CAC) (since 2016), of the Section Committee of Informatics of the Academia Europaea (since 2015), a member of the Academia Europaea: The Academy of Europe, Informatics Section (since 2013).
    He is recipient of the ACM Distinguished Scientist award (2009) and the IBM Faculty Award (2012).
    He is an Associate Editor of IEEE Transactions on Services Computing, ACM Transactions on the Web, and ACM Transactions on Internet Technology and on the editorial board of IEEE Internet Computing. He is the Editor-in-Chief of Computing (an SCI-ranked journal of Springer).
\end{IEEEbiography}

\begin{IEEEbiography}[{\includegraphics[width=1.2in,height=1.25in,clip,keepaspectratio]{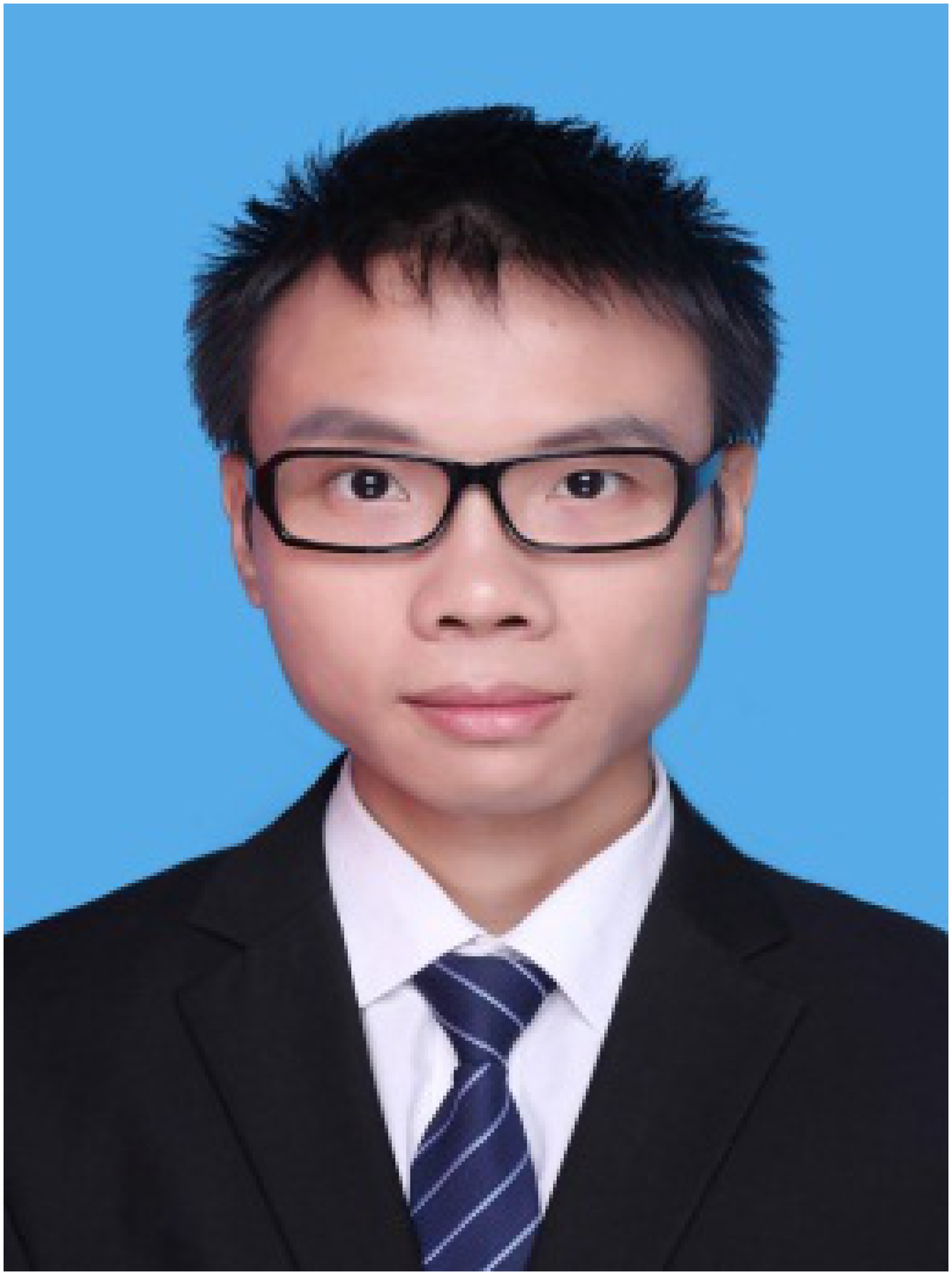}}]{Jiayi Du}
	received his Ph.D., M.S. and B.S. in computer science from Hunan University, China, in 2015, 2010 and 2004.
	He is currently an assistant professor in Central South University of Forest and Technology, China.
	His research interest includes modeling and scheduling for parallel and distributed computing systems, embedded system computing, cloud computing, parallel system reliability, and parallel algorithms.
\end{IEEEbiography}

\end{document}